\pdfoutput=1
\RequirePackage[l2tabu,orthodox]{nag}
\documentclass[USenglish]{article}
\usepackage[a4paper, margin=1in]{geometry}

\usepackage[utf8]{inputenc}
\usepackage[T1]{fontenc}
\usepackage{microtype}
\usepackage{lmodern}

\usepackage{amsmath}
\usepackage{amssymb}

\usepackage[within=none,figurename=Fig.,labelsep=period]{caption} 
\usepackage{subcaption}

\usepackage{hyperref}

\usepackage{authblk}

\title{Finding Detours is Fixed-parameter Tractable\normalfont\footnote{%
    Extended abstract appears at ICALP 2017.
    Most of this work was done while the authors were visiting the Simons Institute for the Theory of Computing.
    IB is supported by NSF grant CCF-1319987.
    RC is supported by ERC grant PARAMTIGHT (No. 280152).%
}}

\author[1]{Ivona Bez\'akov\'a}
\author[2]{Radu Curticapean}
\author[3]{Holger Dell}
\author[4]{Fedor~V.~Fomin}
\affil[1]{Department of Computer Science, Rochester Institute of Technology, Rochester, NY, U.S.A.,
  \texttt{ib@cs.rit.edu}}
\affil[2]{Institute for Computer Science and Control, Hungarian Academy of Sciences (MTA SZTAKI), Budapest, Hungary,
  \texttt{radu.curticapean@gmail.com}}
\affil[3]{Saarland University and Cluster of Excellence, MMCI, Saarbr\"ucken, Germany,
  \texttt{hdell@mmci.uni-saarland.de}}
\affil[4]{University of Bergen, Bergen, Norway,
  \texttt{fomin@ii.uib.no}}

\usepackage{xspace}

\newcommand{\probKPath}{\textsc{Longest Path}\xspace}
\newcommand{\probEXKPath}{\textsc{Exact Path}\xspace}
\newcommand{\probEXDet}{\textsc{Exact Detour}\xspace}
\newcommand{\probLongDet}{\textsc{Longest Detour}\xspace}
\newcommand{\dist}{d}

\newcommand{\defparproblem}[4]{
  \vspace{1mm}
\noindent\fbox{
  \begin{minipage}{0.96\textwidth}
  \begin{tabular*}{\textwidth}{@{\extracolsep{\fill}}lr} #1  &
    {\textbf{Parameter:}} #3 \\ \end{tabular*}
  {\textbf{Input:}} #2  \\
  {\textbf{Task:}} #4
  \end{minipage}
  }
  \vspace{1mm}
}

\newcommand{\op}[1]{\ensuremath{\operatorname{#1}}}
\newcommand{\poly}{\op{poly}}
\newcommand{\tw}{\op{tw}}

\newcommand{\Oh}{{O}}

\newcommand{\N}{\mathbf{N}}

\newcommand{\Q}{\mathbf{Q}}

\global\long\def\relevantpart#1#2#3{#1_{#2,#3}}
\global\long\def\subdivtetra#1{K_{4}^{\geq#1}}

\usepackage[amsmath,hyperref,thmmarks]{ntheorem}

\theoremseparator{.}
\newtheorem{theorem}{Theorem}
\newtheorem{lemma}[theorem]{Lemma}
\newtheorem{proposition}[theorem]{Proposition}
\newtheorem{corollary}[theorem]{Corollary}

\theoremstyle{plain}
\theorembodyfont{\upshape}

\newtheorem{definition}[theorem]{Definition}

\theoremstyle{nonumberplain}
\theoremheaderfont{\itshape}
\theorembodyfont{\upshape}
\theoremsymbol{\ensuremath{_\blacksquare}}
\newtheorem{proof}{Proof}

\usepackage{mathtools}

\DeclarePairedDelimiter\paren{\lparen}{\rparen}
\DeclarePairedDelimiter\abs{\lvert}{\rvert}
\DeclarePairedDelimiter\set{\{}{\}}

\newenvironment{algor}[3]{%
  \bigskip\hrule\nopagebreak[4]\smallskip\nopagebreak[4]
\noindent\textbf{Algorithm #1} ({\itshape#2\/}) {\itshape#3}
\begin{description}
\smallskip
\setlength{\itemsep}{0pt}\leftmargin=0pt}{%
\end{description}\nopagebreak[4]
  \smallskip\hrule\nopagebreak[4]\medskip\nopagebreak[4]}

\usepackage{tikz}
\usetikzlibrary{calc}

\begin{document}

\maketitle

\begin{abstract}
We consider the following natural ``above guarantee'' parameterization of the classical \probKPath problem:
For given vertices $s$ and $t$ of a graph $G$, and an integer~$k$, the problem \probLongDet asks for an $(s,t)$-path in $G$ that is at least~$k$ longer than a shortest $(s,t)$-path.
Using insights into structural graph theory, we prove that \probLongDet is fixed-parameter tractable (FPT) on undirected graphs and actually even admits a single-exponential algorithm, that is, one of running time $\exp(O(k)) \cdot \poly(n)$. This matches (up to the base of the exponential) the best algorithms for finding a path of length at least $k$.

Furthermore, we study the related problem \probEXDet that asks whether a graph $G$ contains an
$(s,t)$-path that is exactly $k$ longer than a shortest $(s,t)$-path.
For this problem, we obtain a randomized algorithm with running time about~$2.746^k\cdot\poly(n)$, and a
deterministic algorithm with running time about $6.745^{k}\cdot\poly(n)$, showing that this problem is FPT as well. Our algorithms for \probEXDet apply to both undirected and directed graphs.
\end{abstract}

\section{Introduction}

The \probKPath problem asks, given an undirected $n$-vertex graph $G$ and an integer~$k$, to decide whether~$G$ contains a path of length at least $k$, that is, a self-avoiding walk with at least~$k$ edges.
This problem is a natural generalization of the classical NP-complete
\textsc{Hamiltonian Path} problem, and the parameterized complexity community
has paid exceptional attention to it.
For instance, Monien~\cite{Monien85} and Bodlaender \cite{Bodlaender93a} showed
\emph{avant la lettre} that \probKPath is fixed-parameter tractable with
parameter~$k$ and admits algorithms with running time $2^{\Oh(k\log
  k)}n^{\Oh(1)}$.
This led Papadimitriou and Yannakakis
\cite{PapadimitriouY96} to conjecture that \probKPath is solvable in polynomial
time for $k=\log{n}$, and indeed, this conjecture was resolved in a seminal paper of Alon,
Yuster, and Zwick~\cite{AlonYZ}, who introduced the method of {color coding} and derived from it the first algorithm with running time $2^{\Oh(k)}n$.  Since this breakthrough of Alon et al.~\cite{AlonYZ}, the problem \probKPath occupied a central place in parameterized algorithmics, and several novel
approaches were developed in order to  reduce the base of the exponent in  the
running time
\cite{HuffnerWZ08,KneisMRR06,ChenLSZ07,ChenKLMR09,Koutis08,Williams09,FominLS14,FominLS14,Bjorklund2017119}.
We refer to  two review articles in Communications of ACM
\cite{FominK13,KoutisW16} as well as to the textbook
\cite[Chapter~10]{cygan2015parameterized} for an extensive overview of
parameterized   algorithms for \probKPath.  Let us however note that the fastest known randomized
algorithm for \probKPath is due to Björklund et al.~\cite{Bjorklund2017119} and runs
in time $1.657^k \cdot n^{\Oh(1)} $, whereas the fastest known deterministic algorithm
is due to Zehavi~\cite{Zehavi14} and runs in time $2.597^k \cdot n^{\Oh(1)}$.

In the present paper, we study the problem \probKPath from the perspective of an
``above guarantee'' parameterization that can attain small values even for long
paths: For a  pair of vertices $s,t\in V(G)$, we use
$\dist_G(s,t)$ to denote the distance, that is, the length of a shortest path
from $s$ to $t$. We then ask for an $(s,t)$-path of length at least
$d_G(s,t)+k$, and we parameterize by this offset $k$ rather than the actual
length of the path to obtain the problem \probLongDet. In other words, the first
$d_G(s,t)$ steps on a path sought by \probLongDet are complimentary and will not
be counted towards the parameter value. This reflects the fact that shortest
paths can be found in polynomial time and could (somewhat embarrassingly) be
much better solutions for \probKPath than the paths of logarithmic length found
by algorithms that parameterize by the path length.

We study two variants of the detour problem, one asking for a detour of length at least~$k$, and another asking for a detour of length exactly~$k$.

\smallskip

\defparproblem{\probLongDet}{Graph $G$, vertices $s,t\in V(G)$, and integer
  $k$.}{$k$}{Decide whether there is an $(s,t)$-path in $G$ of length at least
  $\dist_G(s,t) +k$.}

\smallskip

\defparproblem{\probEXDet}{Graph $G$, vertices $s,t\in V(G)$, and integer
  $k$.}{$k$}{Decide whether there is an $(s,t)$-path in $G$ of length exactly
  $\dist_G(s,t) +k$.}

\smallskip

Our parameterization above the length of a shortest path is a new example in the
general paradigm of ``above guarantee'' parameterizations, which was introduced
by Mahajan and Raman~\cite{MahajanR99}.
Their approach was successfully applied to various problems, such as finding
independent sets in planar graphs (where an independent set of size at least
$\frac n4$ is guaranteed to exist by the Four Color Theorem), or the maximum cut
problem, see e.g.
\cite{DBLP:journals/algorithmica/AlonGKSY11,CrowstonJMPRS13,GutinIMY12,DBLP:journals/mst/GutinKLM11,MahajanRS09}.

\subsection*{Our results}
We show the following tractability results for \probLongDet and \probEXDet:
\begin{itemize}
  \item
    \probLongDet is fixed-parameter tractable (FPT) on undirected graphs.  The
    running time of our algorithm is single-exponential, i.e., it is of the type $2^{\Oh(k)}\cdot n^{\Oh(1)}$
    and thus asymptotically matches the running time of algorithms for
    \probKPath. Our approach requires a non-trivial argument in graph structure
    theory to obtain the single-exponential algorithm; a mere
    FPT-algorithm could be achieved with somewhat less effort. It should also be
    noted that a straightforward reduction rules out a running time of
    $2^{o(k)}\cdot n^{\Oh(1)}$ unless the exponential-time hypothesis of
    Impagliazzo and Paturi~\cite{IP01} fails.
  \item
    \probEXDet is FPT on directed and undirected graphs. Actually, we give a
    polynomial-time Turing reduction from  \probEXDet to the standard parameterization of \probKPath,
    in which we ask on input $u,v$ and $k \in \N$ whether there is a $(u,v)$-path
    of length~$k$. This reduction only makes queries to instances with
    parameter at most~${2k+1}$.
    Pipelined with the fastest known algorithms for  \probKPath mentioned above,
    this implies that  \probEXDet admits  a bounded-error randomized algorithm
    with running time~$2.746^kn^{\Oh(1)}$, and a deterministic algorithm with
    running time~${6.745^{k} n^{\Oh(1)}}$.
\end{itemize}

By a self-reducibility argument, we also show how to construct the required paths rather than just detect their existence. This reduction incurs only polynomial overhead.

\subsection*{Techniques}
The main idea behind the algorithm for  \probLongDet is the following
combinatorial theorem, which shows the existence of specific large planar minors in large-treewidth graphs while circumventing the full machinery used in the Excluded Grid Theorem \cite{RobertsonS-V}. Although the Excluded Grid Theorem already shows that graphs of sufficiently large treewidth contain arbitrary fixed planar graphs, resorting to more basic techniques allows us to show that linear treewidth suffices for our specific cases.
More specifically, we show that there
exists a global constant $c\in\N$ such that every graph of treewidth at least
$c\cdot k$ contains as a subgraph a copy of a graph $\subdivtetra k$, which is any graph obtained from the complete graph~$K_4$ by replacing every edge by a path with at least~$k$ edges.  The proof of this result is based on the
structural theorems of Leaf and Seymour \cite{LeafS15} and Raymond and Thilikos \cite{RaymondT16}.

With the combinatorial theorem at hand, we implement the following win/win
approach: If the treewidth of the input graph is less than $c \cdot k$, we use
known algorithms \cite{rank-treewidth,FominLS14} to solve the problem in single-exponential time. Otherwise
the treewidth of the input graph is at least $c \cdot k$ and there must be a
$\subdivtetra k$, which we use to argue that any path visiting the same two-connected component as $\subdivtetra k$ can be
prolonged by rerouting it through $\subdivtetra k$. To this end, we set up a fixed system of linear inequalities corresponding to the possible paths in $\subdivtetra k$ such that rerouting is possible if and only if the system is unsatisfiable. We then verify the unsatisfiability of this fixed system by means of a computer-aided proof (more specifically, a linear programming solver). From LP duality, we also obtain a short certificate for the unsatisfiability, which we include in the appendix.

The algorithm for \probEXDet is based on the following idea.  We run
breadth-first search (BFS) from vertex $v$ to vertex $u$. Then, for every
$(u,v)$-path $P$ of length $\dist_G(u,v)+k$, all but at most $k$ levels of the
BFS-tree contain exactly one vertex of $P$. Using this property, we are able to
devise a dynamic programming algorithm for \probEXDet, provided it is given access to an
oracle for \probKPath.

\medskip
The remaining part of the paper is organized as follows:
\S\ref{sec:Prelim} contains definitions and preliminary results used in the technical part of the paper.
In \S\ref{sec:winwin}, we give an algorithm for   \probLongDet while \S\ref{sec:DP} is devoted to \probEXDet.
We provide a search-to-decision reduction for \probLongDet and \probEXDet
in \S\ref{sec:searchtodecision}.
In \S\ref{sec: LP unsatisfiable}, we give short certificates for the unsatisfiability of the linear programs from \S\ref{sec:winwin}.

\section{Preliminaries}\label{sec:Prelim}
We consider graphs~$G$ to be undirected, and we denote by~$uv$ an undirected
edge joining vertices $u,v\in V(G)$.
A \emph{path} is a self-avoiding walk in~$G$; the \emph{length} of the path is
its number of edges.
An $(s,t)$-path for $s,t\in V(G)$ is a path that starts at~$s$ and ends at~$t$.
We allow paths to have length~$0$, in which case $s=t$ holds.
For a vertex set $X\subseteq V(G)$, denote by~$G[X]$ the subgraph induced by $X$.

\subparagraph*{Tree decompositions.}
A {\em tree decomposition} $\mathcal{T}$ of a graph $G$ is a pair
$(T,\{X_t\}_{t\in V(T)})$, where~$T$ is a tree in which every node $t$ is
assigned a vertex subset $X_t\subseteq V(G)$, called a bag, such that the
following three conditions hold:
\begin{description}
\item[(T1)]
  Every vertex of $G$ is in at least one bag, that is,
  $V(G)=\bigcup_{t\in V(T)} X_t$.
\item[(T2)]
  For every $uv\in E(G)$, there exists a node $t\in V(T)$ such that $X_t$
  contains both~$u$ and~$v$.
\item[(T3)]
  For every  $u\in V(G)$, the set $T_u  $   of all nodes of $T$ whose
  corresponding  bags contain~$u$, induces a connected subtree of $T$.
\end{description}
The {\em width} of the tree decomposition $\mathcal{T}$ is the integer
$\max_{t\in V(T)} |X_t| - 1$, that is, the size of its largest bag minus~$1$.
The {{\em treewidth}} of a graph~$G$, denoted by~$\tw(G)$, is the smallest
possible width that a tree decomposition of~$G$ can have.

We will need the following algorithmic results about treewidth.
\begin{proposition}[\cite{BodlaenderDDFLP16}]\label{prop:twcomp}
There is a $2^{\Oh(k)} \cdot n$ time algorithm that, given a graph $G$ and an
integer~$k$, either outputs a tree decomposition of width at most $5k+4$, or
correctly decides that $\tw(G)>k$.
\end{proposition}
\begin{proposition}[\cite{rank-treewidth,FominLS14}]\label{prop:twalg}
  There is an algorithm with running time $2^{\Oh(\tw(G))} \cdot n^{\Oh(1)} $ that
  computes a longest
  path between two given vertices of a given graph.
\end{proposition}
Let us note that  the running time of Proposition~\ref{prop:twalg}  can be improved to  $2^{\Oh(\tw(G))} \cdot n$  by making use of the matroid-based approach from \cite{FominLS14}.

Our main theorem is based on graph minors, and we introduce some notation here.

\begin{definition}\label{def: minor-models}
A \emph{topological minor model} of $H$ in $G$ is a pair
of functions $(f,p)$ with $f:V(H)\to V(G)$ and $p:E(H)\to2^{E(G)}$
such that
\begin{enumerate}
\item  $f$ is injective, and
\item  for every edge $uv\in E(H)$, the graph~$G[p(uv)]$ is a path from
$f(u)$ to $f(v)$ in $G$, and
\item  for edges $e,g\in E(H)$ with~${e\neq g}$, the paths~$G[p(e)]$ and
$G[p(g)]$ intersect only in endpoints or not at all.
\end{enumerate}

The graph $T$ \emph{induced by} the topological minor model $(f,p)$ is the
subgraph of~$G$ that consists of the union of all paths $G[p(uv)]$ over all
$uv\in E(H)$.
The vertices in $f(V(H))$ are the \emph{branch vertices} of $T$, and $G[p(e)]$
\emph{realizes} the edge $e$ in $T$.
\end{definition}

\section{Win/Win algorithm for Longest Detour}\label{sec:winwin}

Throughout this section, let $G$ be an undirected graph
with $n$ vertices and $m$ edges, let $s,t\in V(G)$ and~$k\in\mathbb{N}$.
We wish to decide in time $2^{\Oh(k)}\cdot n^{\Oh(1)}$
whether $G$ contains an $(s,t)$-path of length at least $d_{G}(s,t)+k$.
To avoid trivialities, we assume without loss of generality that~$G$ is
connected and $s\neq t$ holds.
Moreover, we can safely remove vertices $v$ that are not part of any $(s,t)$-path.
\begin{definition}
Let $G$ be a graph and let $s,t\in V(G)$.
The \emph{$(s,t)$-relevant part of $G$} is the graph induced by all vertices
contained in some~$(s,t)$-path.
We denote it by $\relevantpart Gst$.
\end{definition}

The graph $\relevantpart Gst$ can be computed efficiently from the block-cut
tree of~$G$.
Recall that the \emph{block-cut tree} of a connected graph~$G$ is a tree where
each vertex corresponds to a \emph{block}, that is, a maximal biconnected
component $B\subseteq V(G)$, or to a \emph{cut vertex}, that is, a vertex whose
removal disconnects the graph.
A block~$B$ and a cut vertex~$v$ are adjacent in the block-cut tree if and only
if there is a block~$B'$ such that~$B\cap B'=\set{v}$.
\begin{lemma}\label{lem: block cut relevant}
Let $B_{s}$ and $B_{t}$ denote the blocks of $G$ that contain $s$
and $t$, respectively.
Furthermore, let $P$ be the unique $(B_s,B_t)$-path in the block-cut tree of~$G$.
Then $\relevantpart Gst$ is the graph induced by the union of all blocks visited
by~$P$.
\end{lemma}
\begin{proof}
  Let $v\in G_{s,t}$.
  Then there is an $(s,t)$-path that contains~$v$; in particular, there is an
  $(s,v)$-path $p_1$ and a $(v,t)$-path $p_2$ such that $p_1$ and $p_2$ are
  internally vertex disjoint.
  If $v$ was not in one of the blocks visited by~$P$, it would be hidden behind
  a cut vertex and $p_1$ and $p_2$ would have to intersect in the cut vertex;
  therefore, $v$ is contained in one of the blocks visited by~$P$.

  For the other direction, let $v$ be a vertex contained in a block~$B$ visited
  by~$P$.
  Suppose that $u$ is the cut vertex preceding~$B$ in~$P$ (or $u=s$ in case
  $B=B_s$) and $w$ is the cut vertex following~$B$ in~$P$ (or $w=t$ in case
  $B=B_t$).
  Then $u\neq w$ holds, and there is an $(s,u)$-path and a $(w,t)$-path that are
  vertex-disjoint.
  Since $B$ is biconnected, there are paths from $u$ to $v$ and from $v$ to $w$
  that are internally vertex-disjoint.
  Combined, these path segments yield an $(s,t)$-path that visits~$v$.
\end{proof}

We formulate an immediate implication of Lemma~\ref{lem: block cut relevant}
that will be useful later.
\begin{corollary}
\label{cor: relevant-part has block-cut path}
The block-cut tree of $\relevantpart Gst$ is a $(B_s,B_t)$-path.
\end{corollary}
Hopcroft and Tarjan~\cite{HopcroftTarjan73} proved that the block-cut tree of a
graph can be computed in linear time using DFS.
Hence we obtain an algorithm for computing~$G_{s,t}$ from~$G$.
\begin{corollary}\label{cor: compute-relevant}
There is a linear-time algorithm that computes~$\relevantpart Gst$ from~$G$.
\end{corollary}

\subsection{The algorithm}

By definition, the graph $\relevantpart Gst$ contains the same set of
$(s,t)$-paths as~$G$.
Our algorithm for \probLongDet establishes a ``win/win'' situation as follows:
We prove that, if the treewidth of $\relevantpart Gst$ is ``sufficiently
large'', then $(G,s,t,k)$ is a $\mathtt{YES}$-instance of \probLongDet.
Otherwise the treewidth is small, and we use a known treewidth-based dynamic
programming algorithm for computing the longest~$(s,t)$-path.
Hence the algorithm builds upon the following subroutines:
\begin{enumerate}
  \item The algorithm from Corollary~\ref{cor: compute-relevant},
    computing the relevant part~$\relevantpart Gst$ of~$G$ in time
    $\Oh(n+m)$.
  \item \textsc{Compute Treewidth}$(G,w)$ from Proposition~\ref{prop:twcomp}, which is given~$G$
    and $w\in\mathbb{N}$ as input, and either constructs a
    tree-decomposition~$T$ of~$G$ whose width is bounded by $5w+4$, or outputs
    $\mathtt{LARGE}$.
    If the algorithm outputs $\mathtt{LARGE}$, then $\tw(G)>w$ holds.
    The running time is $2^{\Oh(w)}\cdot
    n$.
  \item $\probKPath(G,T,s,t)$ from Proposition~\ref{prop:twalg}, which is given
    $G,s,t$
    and additionally a tree-decomposition~$T$ of~$G$, and outputs a longest $(s,t)$-path
    in~$G$. The running time is $2^{\Oh(w)}\cdot n^{\Oh(1)}$,
    where~$w$ denotes the width of~$T$.
\end{enumerate}

We now formalize what we mean by ``sufficiently large'' treewidth.
\begin{definition}
A function $f:\mathbb{N}\to\mathbb{N}$ is \emph{detour-enforcing}
if, for all $k\in\mathbb{N}$ and all graphs~$G$ with vertices~$s$ and $t$, the
following implication holds:
If $\tw(\relevantpart Gst)>f(k)$, then $G$ contains an $(s,t)$-path
of length at least $d_{G}(s,t)+k$.
\end{definition}
\begin{theorem}\label{thm: detour-fn-exists-strong}
  The function $f:k\mapsto 32 k + 2$ is detour-enforcing.
\end{theorem}

We defer the proof of this theorem to the next section, and instead state
Algorithm D, which uses~$f$ to solve \probLongDet.
Algorithm D turns out to be an FPT-algorithm already when any detour-enforcing
function~$f$ is known (as long as it is polynomial-time computable),
and it becomes faster when detour-enforcing~$f$ of slower growth are used.

\begin{algor}{D}{\probLongDet}{Given $(G,s,t,k)$, this algorithm decides whether
    the graph~$G$ contains an $(s,t)$-path of length at least $\dist_G(s,t)+k$.}

  \item[D1] (Restrict to relevant part)
    Compute $G_{s,t}$ using Corollary~\ref{cor: compute-relevant}.
  \item[D2] (Compute shortest path)
    Compute the distance $d$ between $s$ and $t$ in $G_{s,t}$.
  \item[D3] (Compute tree-decomposition)
    Call $\textsc{Compute Treewidth}(G_{s,t},f(k))$.
  \begin{description}
    \item[D3a] (Small treewidth) If the subroutine found a
      tree-decomposition~$T$ of width at most~$f(k)$, call
      $\probKPath(G_{s,t},T,s,t)$.
      Output \texttt{YES} if there is an $(s,t)$-path of length at least
      $d+k$, otherwise output \texttt{NO}.
    \item[D3b] (Large treewidth) If the subroutine returned \texttt{LARGE},
      output $\mathtt{YES}$.
  \end{description}
\end{algor}

We prove the running time and correctness of Algorithm~D.
\begin{lemma}\label{lem: algo-from-detour-fn}
  For every polynomial-time computable detour-enforcing function
  $f:\mathbb{N}\to\mathbb{N}$, Algorithm~D solves \probLongDet in time
  $2^{\Oh(f(k))}\cdot n^{\Oh(1)}$.
\end{lemma}
\begin{proof}
Using \textsc{Compute Treewidth}$(\relevantpart Gst,f(k))$, we first
determine in time $2^{\Oh(f(k))}\cdot n^{\Oh(1)}$
whether $\tw(\relevantpart Gst)\leq f(k)$.
\begin{itemize}
  \item If $\tw(\relevantpart Gst)\leq f(k)$, then \textsc{Compute Treewidth}
    yields a tree decomposition $T$ of $\relevantpart Gst$ whose width
    is bounded by $5\cdot f(k)+4$.
    We invoke the algorithm for \probKPath to compute a longest $(s,t)$-path in
    $G_{s,t}$ and we output $\mathtt{YES}$ if and only if its length is at least
    $d(s,t)+k$.
    Since the $(s,t)$-paths in $G$ are precisely the $(s,t)$-paths in $G_{s,t}$,
    this answer is correct.
    The running time of this step is at most $2^{\Oh(f(k))}\cdot
    n^{\Oh(1)}$.
  \item If $\tw(\relevantpart Gst)>f(k)$, we output $\mathtt{YES}$. Since
    $f$ is detour-enforcing, the graph $G$ indeed contains an $(s,t)$-path
    of length at least $d_{G}(s,t)+k$.
\end{itemize}
We conclude that Algorithm~D is correct and observe also that its running time is
bounded by ${2^{\Oh(f(k))}\cdot n^{\Oh(1)}}$.
\end{proof}

Theorem~\ref{thm: detour-fn-exists-strong} and Lemma~\ref{lem:
algo-from-detour-fn} imply a $2^{\Oh(k)} \cdot n^{\Oh(1)}$ time algorithm for
\probLongDet.

\subsection{Overview of the proof of Theorem~\ref{thm: detour-fn-exists-strong}}

In our proof of Theorem~\ref{thm: detour-fn-exists-strong}, large subdivisions
of $K_{4}$ play an important role. Intuitively speaking, a sufficiently
large subdivision of $K_{4}$ in $\relevantpart Gst$ allows us to
route some $(s,t)$-path through it and then exhibit a long detour within
that subdivision.

\begin{definition}
For $k\in\mathbb{N}$, a graph $F$ is a $\subdivtetra k$ if it can
be obtained by subdividing each edge of~$K_{4}$ at least~$k$ times.
Please note that the numbers of subdivisions do not need to agree for different
edges.
\end{definition}
We show in Section~\ref{sec: K4 rerouting} that graphs $G$ containing $\subdivtetra k$
subgraphs in $\relevantpart Gst$ have $k$-detours.
\begin{lemma}
\label{lem: K4 implies detour}Let $G$ be a graph and $k\in\mathbb{N}$.
If $\relevantpart Gst$ contains a $\subdivtetra k$ subgraph, then
$G$ contains an $(s,t)$-path of length at least $d_{G}(s,t)+k$.
\end{lemma}
Since the graph obtained by subdividing each edge of $K_{4}$ exactly
$k$ times is a planar graph on $\Oh(k)$ vertices, the Excluded
Grid Theorem yields a function
$f:\mathbb{N}\to\mathbb{N}$ such that every graph of treewidth at
least $f(k)$ contains some $\subdivtetra k$ minor. Furthermore,
since every $\subdivtetra k$ has maximum degree $3$, this actually
shows that $G$ contains some $\subdivtetra k$ as a\emph{ subgraph}.
Thus, Lemma~\ref{lem: K4 implies detour} implies that $f$ is detour-enforcing,
and a proof of this lemma immediately implies a weak version of Theorem
\ref{thm: detour-fn-exists-strong}.

By recent improvements on the Excluded Grid Theorem~\cite{ChekuriC13,Chuzhoy16}, the function $f$
above is at most a polynomial.
However, even equipped with this deep result we cannot obtain a
single-exponential algorithm for \probLongDet using the approach of
Lemma~\ref{lem: algo-from-detour-fn}: It would require~$f$ to be linear.
In fact, excluding grids is too strong a requirement for us, since every
function~$f$ obtained as a corollary of the full Excluded Grid Theorem must be
super-linear~\cite{RobSeymT94}.
We circumvent the use of the Excluded Grid Theorem and prove the following lemma
from more basic principles.
\begin{lemma}\label{lem: tw implies K4}%
  For graphs $G$ and $k\in\mathbb{N}$, if
  $\tw(G)\geq 32k+2$, then $G$ contains a $\subdivtetra k$ subgraph.
\end{lemma}
Together, Lemmas~\ref{lem: tw implies K4} and \ref{lem: K4 implies detour}
imply Theorem~\ref{thm: detour-fn-exists-strong}.

\begin{proof}[of Theorem~\ref{thm: detour-fn-exists-strong}]
  Let $G$ and $s,t\in V(G)$ and $k\in\mathbb{N}$ be such that $\tw(\relevantpart
  Gst)>f(k)$.
  By Lemma~\ref{lem: tw implies K4}, the graph $\relevantpart Gst$ contains a
  $\subdivtetra k$ subgraph, so Lemma~\ref{lem: K4 implies detour} implies that
  $G$ contains an $(s,t)$-path of length $d_{G}(s,t)+k$.
  This shows that $f$ is indeed detour-enforcing.
\end{proof}

\subsection{Proof of Lemma~\ref{lem: K4 implies detour}: Rerouting in subdivided
  tetrahedra}

\label{sec: K4 rerouting}

Let $(G,s,t,k)$ be an instance for \probLongDet such that
$\relevantpart Gst$ contains a $\subdivtetra k$ subgraph~$M$.
We want to prove that $\relevantpart Gst$ has a path of length at
least~$\dist_G(s,t)+k$; in fact, we construct the desired detour entirely in the
subgraph~$M$, for which reason we first need to route some $(s,t)$-path
through~$M$.
\begin{lemma}\label{lemma:route to K4}
  There are two distinct vertices $u,v\in V(M)$ and two vertex-disjoint
  paths~$P_s$ and~$P_t$ in~$G$ such that $P_s$ is an $(s,u)$-path, $P_t$ is a
  $(v,t)$-path, and they only intersect with~$V(M)$ at~$u$ and~$v$.
\end{lemma}
The proof of this lemma uses the fact that every block in the block-cut tree is
biconnected.
\begin{proof}
  Since $\subdivtetra k$ is biconnected, $M$ is contained in a single block $C$
  of~$\relevantpart Gst$.
  By Corollary~\ref{cor: relevant-part has block-cut path}, the block-cut tree
  of $\relevantpart Gst$ is a path.
  Let $s'$ be the cut vertex preceding~$C$ in this block-cut tree (or $s'=s$ if
  $C$ is the first block) and let $t'$ be the cut vertex following~$C$ in the
  tree (or $t'=t$ if $C$ is the last block).
  Then clearly $s',t'\in C$.

  By the properties of the block-cut tree, there is an $(s,s')$-path~$p_s$ and a
  $(t',t)$-path~$p_t$, the two paths are vertex disjoint, and they intersect~$C$
  only in~$s'$ and~$t'$, respectively.
  We let $p_s$ be the first segment of~$P_s$ and $p_t$ be the last segment
  of~$P_t$.
  It remains to complete~$P_s$ and~$P_t$ within~$C$ using two disjoint paths that lead
  to~$M$.
  Since~$C$ is biconnected, there are two vertex-disjoint paths
  from~$\set{s',t'}$ to $V(M)$.
  Moreover, both paths can be shortened if they intersect~$V(M)$ more than once.
  Hence we have an $(s',u)$-path~$p_1$ for some $u\in V(M)$ and a
  $(v,t')$-path~$p_2$ for some $v\in V(M)$ with the property that~$p_1$
  and~$p_2$ are disjoint and their internal vertices avoid~$V(M)$.

  We concatenate the paths $p_s$ and $p_1$ to obtain~$P_s$ and the paths~$p_2$
  and~$p_t$ to obtain~$P_t$.
\end{proof}
Next we show that every $\subdivtetra k$-graph~$M$ contains long detours.
\begin{lemma}\label{lem: detour-in-K4}
  Let $M$ be a $\subdivtetra k$-graph.
  For every two distinct vertices $u,v\in V(M)$, there is~a $(u,v)$-path of
  length at least $d_{M}(u,v)+k$ in $M$.
\end{lemma}

The proof idea is to distinguish cases depending on where $u,v$
lie in~$M$ relative to each other. For each case, we can exhaustively list
all $(u,v)$-paths (see Figure~\ref{fig: K4 uv}).
We do not quite know the lengths of these paths, but we do know that each
has length at least~$d_M(u,v)$; moreover, each $(b_i,b_j)$-path
in~$M$ for two distinct degree-$3$ vertices~$b_i$ and $b_j$ has length at
least~$k$, since we subdivided~$K_4$ at least~$k$ times.
The claim of Lemma~\ref{lem: detour-in-K4} is that one of the $(u,v)$-paths must
have length at least $d_M(u,v)$.
To prove this, we set up a linear program where the variables are $d_M(u,v)$,
$k$, and the various path segment lengths; its infeasibility
informs us that indeed a path that is longer by~$k$ must exist.

\definecolor{carmine}{HTML}{A20021}

\tikzstyle{vtx}=[circle,fill=black!70,inner sep=0pt,minimum width=4pt]
\tikzstyle{tiny}=[vtx,fill=black!70,minimum width=2.0pt,draw=none]
\tikzstyle{uv}=[vtx,rectangle,minimum width=4pt,minimum height=4pt,draw=black,thick,fill=carmine]

\tikzstyle{uvpath}=[line width=2pt,draw=carmine,cap=rect]

\def\largesidelength{3.0}
\def\smallsidelength{1.5}
\def\pathlabeldistance{.3}

\pgfdeclarelayer{foreground}
\pgfsetlayers{main,foreground}

\def\tetrasetup{
  \begin{pgfonlayer}{foreground}
    \node[vtx] (b1) at (0,0) {};
    \node[vtx] (b2) at (\sidelength,0) {};
    \node[vtx] (b4) at (\sidelength/2,{\sidelength*sqrt(3)/2-\sidelength/sqrt(3)}) {};
    \node[vtx] (b3) at (\sidelength/2,{\sidelength*sqrt(3)/2}) {};
  \end{pgfonlayer}

  \draw[cap=rect] (b1) -- (b2);
  \draw[cap=rect] (b1) -- (b3);
  \draw[cap=rect] (b1) -- (b4);
  \draw[cap=rect] (b2) -- (b3);
  \draw[cap=rect] (b2) -- (b4);
  \draw[cap=rect] (b3) -- (b4);

  \path[use as bounding box]
  ($(b1)-(.2,0.5)$) -- ($(b2)+(.2,-0)$) -- ($(b3)+(0,.6)$);
}
\def\uvdraw{
  \begin{pgfonlayer}{foreground}
    \node[uv] at (u) {};
    \node[uv] at (v) {};
  \end{pgfonlayer}
}
\def\tetrasubdivs{
  \foreach \x in {1,2,3,...,9}
  {
    \node[tiny] at ($ {1-\x/10}*(b1)+\x/10*(b2) $) {};
    \node[tiny] at ($ {1-\x/10}*(b1)+\x/10*(b3) $) {};
    \node[tiny] at ($ {1-\x/10}*(b2)+\x/10*(b3) $) {};
  }
  \foreach \x in {1,2,3,...,5}
  {
    \node[tiny] at ($ {1-\x/6}*(b1)+\x/6*(b4) $) {};
    \node[tiny] at ($ {1-\x/6}*(b2)+\x/6*(b4) $) {};
    \node[tiny] at ($ {1-\x/6}*(b3)+\x/6*(b4) $) {};
  }
}

\begin{figure}[pt]
  \def\uvsetup{
    \coordinate (u) at ($ .7*(b1) + .3*(b2)$) {};
    \coordinate (v) at ($ .4*(b1) + .6*(b2)$) {};
  }
  \begin{subfigure}[c]{\textwidth}
    \begin{minipage}[c]{.29\textwidth}
      \let\sidelength\largesidelength
      \begin{tikzpicture}
        \tetrasetup\uvsetup\tetrasubdivs\uvdraw
        \node[label={[label distance=6pt,anchor=center]-90:{$u$}}] at (u) {};
        \node[label={[label distance=6pt,anchor=center]-90:{$v$}}] at (v) {};

        \node[label={[label distance=6pt,anchor=center]-90:{$b_1$}}] at (b1) {};
        \node[label={[label distance=6pt,anchor=center]-90:{$b_2$}}] at (b2) {};
        \node[label={[label distance=6pt,anchor=center]90:{$b_3$}}]  at (b3) {};
        \node[label={[label distance=6pt,anchor=center]30:{$b_4$}}]  at (b4) {};
      \end{tikzpicture}
    \end{minipage}%
    \hfill%
    \begin{minipage}[c]{.7\textwidth}
      \centering
      \let\sidelength\smallsidelength
      \begin{tikzpicture}
        \tetrasetup
        \uvsetup
        \draw[uvpath] (u) -- (v);
        \uvdraw
        \node at ($(b3)+(0,\pathlabeldistance)$) {$u v$};
      \end{tikzpicture}
      \begin{tikzpicture}
        \tetrasetup
        \uvsetup
        \draw[uvpath] (u) -- (b1) -- (b4) -- (b2) -- (v);
        \uvdraw
        \node at ($(b3)+(0,\pathlabeldistance)$) {$u b_1 b_4 b_2 v$};
      \end{tikzpicture}
      \begin{tikzpicture}
        \tetrasetup
        \uvsetup
        \draw[uvpath] (u) -- (b1) -- (b3) -- (b2) -- (v);
        \uvdraw
        \node at ($(b3)+(0,\pathlabeldistance)$) {$u b_1 b_3 b_2 v$};
      \end{tikzpicture}
      \begin{tikzpicture}
        \tetrasetup
        \uvsetup
        \draw[uvpath] (u) -- (b1) -- (b3) -- (b4) -- (b2) -- (v);
        \uvdraw
        \node at ($(b3)+(0,\pathlabeldistance)$) {$u b_1 b_3 b_4 b_2 v$};
      \end{tikzpicture}
      \begin{tikzpicture}
        \tetrasetup
        \uvsetup
        \draw[uvpath] (u) -- (b1) -- (b4) -- (b3) -- (b2) -- (v);
        \uvdraw
        \node at ($(b3)+(0,\pathlabeldistance)$) {$u b_1 b_4 b_3 b_2 v$};
      \end{tikzpicture}
    \end{minipage}
    \caption{\label{fig: K4 uv case a}%
      $u$ and $v$ lie on the same subdivided edge.}
  \end{subfigure}
  \\[.5cm]
  \def\uvsetup{
    \coordinate (u) at ($ .7*(b1) + .3*(b2)$) {};
    \coordinate (v) at ($ .4*(b1) + .6*(b3)$) {};
  }
  \begin{subfigure}[c]{\textwidth}
    \begin{minipage}[c]{.29\textwidth}
      \let\sidelength\largesidelength
      \begin{tikzpicture}
        \tetrasetup\uvsetup\tetrasubdivs\uvdraw
        \node[label={[label distance=6pt,anchor=center]-90:{$u$}}] at (u) {};
        \node[label={[label distance=6pt,anchor=center]150:{$v$}}] at (v) {};

        \node[label={[label distance=6pt,anchor=center]-90:{$b_1$}}] at (b1) {};
        \node[label={[label distance=6pt,anchor=center]-90:{$b_2$}}] at (b2) {};
        \node[label={[label distance=6pt,anchor=center]90:{$b_3$}}]  at (b3) {};
        \node[label={[label distance=6pt,anchor=center]30:{$b_4$}}]  at (b4) {};
      \end{tikzpicture}
    \end{minipage}%
    \hfill%
    \begin{minipage}[c]{.7\textwidth}
      \centering
      \let\sidelength\smallsidelength
      \begin{tikzpicture}
        \tetrasetup
        \uvsetup
        \draw[uvpath] (u) -- (b1) -- (v);
        \uvdraw
        \node at ($(b3)+(0,\pathlabeldistance)$) {$u b_1 v$};
      \end{tikzpicture}
      \begin{tikzpicture}
        \tetrasetup
        \uvsetup
        \draw[uvpath] (u) -- (b2) -- (b3) -- (v);
        \uvdraw
        \node at ($(b3)+(0,\pathlabeldistance)$) {$u b_2 b_3 v$};
      \end{tikzpicture}
      \begin{tikzpicture}
        \tetrasetup
        \uvsetup
        \draw[uvpath] (u) -- (b1) -- (b4) -- (b3) -- (v);
        \uvdraw
        \node at ($(b3)+(0,\pathlabeldistance)$) {$u b_1 b_4 b_3 v$};
      \end{tikzpicture}
      \begin{tikzpicture}
        \tetrasetup
        \uvsetup
        \draw[uvpath] (u) -- (b2) -- (b4) -- (b3) -- (v);
        \uvdraw
        \node at ($(b3)+(0,\pathlabeldistance)$) {$u b_2 b_4 b_3 v$};
      \end{tikzpicture}
      \\
      \begin{tikzpicture}
        \tetrasetup
        \uvsetup
        \draw[uvpath] (u) -- (b2) -- (b4) -- (b1) -- (v);
        \uvdraw
        \node at ($(b3)+(0,\pathlabeldistance)$) {$u b_2 b_4 b_1 v$};
      \end{tikzpicture}
      \begin{tikzpicture}
        \tetrasetup
        \uvsetup
        \draw[uvpath] (u) -- (b1) -- (b4) -- (b2) -- (b3) -- (v);
        \uvdraw
        \node at ($(b3)+(0,\pathlabeldistance)$) {$u b_1 b_4 b_2 b_3 v$};
      \end{tikzpicture}
      \begin{tikzpicture}
        \tetrasetup
        \uvsetup
        \draw[uvpath] (u) -- (b2) -- (b3) -- (b4) -- (b1) -- (v);
        \uvdraw
        \node at ($(b3)+(0,\pathlabeldistance)$) {$u b_2 b_3 b_4 b_1 v$};
      \end{tikzpicture}
    \end{minipage}
    \caption{\label{fig: K4 uv case b}%
      $u$ and $v$ lie on two adjacent subdivided edges.}
  \end{subfigure}
  \\[.5cm]
  \def\uvsetup{
    \coordinate (u) at ($ .7*(b1) + .3*(b2)$) {};
    \coordinate (v) at ($ 1/3*(b3) + 2/3*(b4)$) {};
  }
  \begin{subfigure}[c]{\textwidth}
    \begin{minipage}[c]{.29\textwidth}
      \let\sidelength\largesidelength
      \begin{tikzpicture}
        \tetrasetup\uvsetup\tetrasubdivs\uvdraw
        \node[label={[label distance=6pt,anchor=center]-90:{$u$}}] at (u) {};
        \node[label={[label distance=6pt,anchor=center]180:{$v$}}] at (v) {};

        \node[label={[label distance=6pt,anchor=center]-90:{$b_1$}}] at (b1) {};
        \node[label={[label distance=6pt,anchor=center]-90:{$b_2$}}] at (b2) {};
        \node[label={[label distance=6pt,anchor=center]90:{$b_3$}}]  at (b3) {};
        \node[label={[label distance=6pt,anchor=center]30:{$b_4$}}]  at (b4) {};
      \end{tikzpicture}
    \end{minipage}%
    \hfill%
    \begin{minipage}[c]{.7\textwidth}
      \centering
      \let\sidelength\smallsidelength
      \begin{tikzpicture}
        \tetrasetup\uvsetup
        \draw[uvpath] (u) -- (b1) -- (b3) -- (v);
        \uvdraw
        \node at ($(b3)+(0,\pathlabeldistance)$) {$u b_1 b_3 v$};
      \end{tikzpicture}
      \begin{tikzpicture}
        \tetrasetup\uvsetup
        \draw[uvpath] (u) -- (b1) -- (b4) -- (v);
        \uvdraw
        \node at ($(b3)+(0,\pathlabeldistance)$) {$u b_1 b_4 v$};
      \end{tikzpicture}
      \begin{tikzpicture}
        \tetrasetup\uvsetup
        \draw[uvpath] (u) -- (b2) -- (b3) -- (v);
        \uvdraw
        \node at ($(b3)+(0,\pathlabeldistance)$) {$u b_2 b_3 v$};
      \end{tikzpicture}
      \begin{tikzpicture}
        \tetrasetup\uvsetup
        \draw[uvpath] (u) -- (b2) -- (b4) -- (v);
        \uvdraw
        \node at ($(b3)+(0,\pathlabeldistance)$) {$u b_2 b_4 v$};
      \end{tikzpicture}
      \\
      \begin{tikzpicture}
        \tetrasetup\uvsetup
        \draw[uvpath] (u) -- (b1) -- (b3) -- (b2) -- (b4) -- (v);
        \uvdraw
        \node at ($(b3)+(0,\pathlabeldistance)$) {$u b_1 b_3 b_2 b_4 v$};
      \end{tikzpicture}
      \begin{tikzpicture}
        \tetrasetup\uvsetup
        \draw[uvpath] (u) -- (b1) -- (b4) -- (b2) -- (b3) -- (v);
        \uvdraw
        \node at ($(b3)+(0,\pathlabeldistance)$) {$u b_1 b_3 b_2 b_4 v$};
      \end{tikzpicture}
      \begin{tikzpicture}
        \tetrasetup\uvsetup
        \draw[uvpath] (u) -- (b2) -- (b3) -- (b1) -- (b4) -- (v);
        \uvdraw
        \node at ($(b3)+(0,\pathlabeldistance)$) {$u b_2 b_3 b_1 b_4 v$};
      \end{tikzpicture}
      \begin{tikzpicture}
        \tetrasetup\uvsetup
        \draw[uvpath] (u) -- (b2) -- (b4) -- (b1) -- (b3) -- (v);
        \uvdraw
        \node at ($(b3)+(0,\pathlabeldistance)$) {$u b_2 b_4 b_1 b_3 v$};
      \end{tikzpicture}
    \end{minipage}
    \caption{\label{fig: K4 uv case c}%
      $u$ and $v$ lie on two non-adjacent subdivided edges.}
  \end{subfigure}
  \caption{\label{fig: K4 uv}%
    \emph{Left:}
    Depicted are all three possible cases for the relative positions of
    vertices~$u$ and~$v$ ($\emph{red squares}$)
    in a subdivided tetrahedron~$\subdivtetra k$ with degree-$3$
    vertices~$b_1,\dots,b_4$ (\emph{gray dots}) and at least~$k=5$ subdivision
    vertices~$(\emph{small gray dots})$.
    \emph{Right:} An exhaustive list of all $(u,v)$-paths (\emph{thick red}); in
    each of the three cases, Lemma~\ref{lem: detour-in-K4} implies that the
    longest among them is at least~$k$ longer than the shortest one.
  }
\end{figure}

\begin{proof}[of Lemma~\ref{lem: detour-in-K4}]
  Let $M$ be a $\subdivtetra k$-graph, let $u,v\in V(M)$, and
  let ${b_{1},\ldots,b_{4}}$ denote the four degree-$3$ vertices
  of~$M$.
  Let $P_u$ be a path in~$M$ that realizes an edge of~$K_4$ and satisfies $u\in
  V(P_u)$, and let $P_v$ be such a path with $v\in V(P_v)$.
  We distinguish three cases as depicted in Figure~\ref{fig: K4 uv}:
  \begin{enumerate}
  \item\label{uv a}
    The two paths are the same, that is, $P_u=P_v$.
  \item\label{uv b}
    The two paths share a degree-$3$ vertex, that is, $\abs{V(P_u)\cap
      V(P_v)}=1$.
  \item\label{uv c}
    The two paths are disjoint, that is, $V(P_{u})\cap V(P_{v})=\emptyset$.
  \end{enumerate}
  By the symmetries of $K_{4}$, this case distinction is exhaustive.
  Since $K_{4}$ has automorphisms that map any edge to any other edge, we can
  further assume that $P_u$ is the path implementing the edge~$b_1b_2$ such that 
  $P_u$ visits the vertices $b_{1}$, $u$, $v$, and $b_{2}$ in this order, see
  Figure~\ref{fig: K4 uv}.

  We exhaustively list the set~$\mathcal{P}$ of $(u,v)$-paths of~$M$ in
  Figure~\ref{fig: K4 uv}.
  Each path is uniquely specified by the sequence of the degree-$3$ vertices it
  visits.
  For example, consider the path $u b_1 b_4 b_2 v$:
  This path consists of the four edge-disjoint segments
  $u b_1$, $b_1b_4$, $b_4b_2$, and $b_2v$; in the example figure, these segments
  have length~$3$, $6$, $6$, and~$4$, respectively.
  Given a path~$P\in\mathcal{P}$, let $S(P)$ be the set of its segments between
  $u$, $v$, and the degree-$3$ vertices.
  For a path or a path segment~$s$, we denote its length by~$\ell(s)$.

  Since~$M$ is a $\subdivtetra k$, every edge of $K_4$ is realized by a path of length at least $k$ in $M$.
  Hence, $\ell(b_i b_j)\ge k$ holds for all $i,j$ with $i\ne j$.
  Moreover, we have $\ell(b_1b_2)=\ell(b_{1}u)+\ell(uv)+\ell(vb_{2})$ in
  case~\ref{uv a}.
  Let $d=d_{M}(u,v)$; clearly $\ell(P) \geq d$ holds for all~$P\in\mathcal P$.
  Our goal is to show that~$M$ has a $(u,v)$-path~$P$ with~$\ell(P)\ge d+k$.
  To this end, we treat~$d$, $k$, and all path segment lengths~$\ell(b_ib_j)$ for $i\ne j$
  and $\ell(b_1 u),\ell(uv),\ell(vb_2)$ as variables in a system of linear inequalities
  and establish that the claim holds if this system is unsatisfiable:
  \begin{align}
    \ell(b_{i}b_{j})
    & \geq k\,,
    &\mbox{for all $i,j$ with $i\ne j$}\,,
    \label{eq: tetrasubdiv}
    \\
    \ell(b_{1}u)+\ell(uv)+\ell(vb_{2})
    & = \ell(b_1b_2)\,,
    \label{eq: uv breakup 2}
    \\
    \sum_{s\in S(P)}\ell(s)
    & \geq d\,,
    &\mbox{for all }P\in\mathcal{P}\,,
    \label{eq: shortest path}
    \\
    \sum_{s\in S(P)}\ell(s)
    & \leq d+k-1\,,
    &\mbox{for all }P\in\mathcal{P}\,.
    \label{eq: no long path}
  \end{align}

  This system has eleven variables.
  Please note that $d$ and $k$ are also considered as variables in our formulation.
  The constraints in~\eqref{eq: tetrasubdiv} express that $M$
  realizes each edge of $K_{4}$ by a path of length at least $k$.
  The constraints in~\eqref{eq: uv breakup 2} express
  that~$u$ and~$v$ lie on the path~$b_1b_2$ and break it up into segments.
  The constraints in~\eqref{eq: shortest path} express that no $(u,v)$-path is shorter than $d$ in length,
  and the constraints in~\eqref{eq: no long path} express that every $(u,v)$-path has
  length strictly less than~$d+k$.
  We prove in the appendix that this linear program is infeasible, and so every
  setting for the variables that satisfies~\eqref{eq:
    tetrasubdiv}--\eqref{eq: shortest path} must violate an inequality
  from~\eqref{eq: no long path}; this means that $M$
  must contain a $(u,v)$-path of length at least~$d+k$ in case~\ref{uv a}.

  The proof is analogous when $u$ and $v$ are on different subdivided edges of
  the subdivided tetrahedron; what changes is the set~$\mathcal P$
  of $(u,v)$-paths as well as the constraints~\eqref{eq: uv breakup 2}.
  In case~\ref{uv b}, we may assume by symmetry that~$P_u$ is the
  $b_{1}b_{2}$-path and $P_v$ is the $b_1b_3$-path of~$M$.
  Then~\eqref{eq: uv breakup 2} is replaced with the
  following constraints.
  \begin{align*}
    \ell(b_{1}u)+\ell(ub_{2}) & = \ell(b_1b_2)\,,\\
    \ell(b_{1}v)+\ell(vb_{3}) & = \ell(b_1b_3)\,.
  \end{align*}
  The resulting linear equation system in case~\ref{uv b} has twelve variables and
  is again infeasible.
  Similarly, in case~\ref{uv c}, the constraints~\eqref{eq: uv breakup 2} are replaced with the following.
  \begin{align*}
    \ell(b_{1}u)+\ell(ub_{2}) & = \ell(b_1b_2)\,,\\
    \ell(b_{3}v)+\ell(vb_{4}) & = \ell(b_3b_4)\,.
  \end{align*}
  This also leads to an infeasible linear equation system with twelve variables.
  
  We conclude that, no matter how~$u$ and~$v$ lie relative to each other in~$M$,
  there is always a $(u,v)$-path that is at least~$k$ longer than a shortest
  one.
\end{proof}

This allows us to conclude Lemma~\ref{lem: K4 implies detour} rather easily.

\begin{proof}[of Lemma~\ref{lem: K4 implies detour}]
  Let $d=\dist_G(s,t)$ be the length of a shortest $(s,t)$-path in~$G$.
  Let $M$ be~a~$\subdivtetra k$ in $\relevantpart Gst$, and let
  $P_s$, $P_t$, $u$, and $v$ be the objects guaranteed by Lemma~\ref{lemma:route
    to K4}.
  Let~$P_{uv}$ be a shortest $(u,v)$-path that only uses edges of~$M$; its
  length is $\dist_M(u,v)$.
  Since the combined path $P_s,P_{uv},P_t$ is an $(s,t)$-path, its length is at
  least~$d$.

  Finally, Lemma~\ref{lem: detour-in-K4} guarantees that there is a $(u,v)$-path
  $Q_{uv}$ in~$M$ whose length is at least $\dist_{M}(u,v)+k$.
  Therefore, the length of the $(s,t)$-path $P_s,Q_{uv},P_t$ satisfies
  \begin{align*}
    \ell(P_{s})+\ell(Q_{uv})+\ell(P_t)
    &\ge \ell(P_{s})+\paren*{\dist_{M}(u,v)+k}+\ell(P_t)
    \\
    &=\ell(P_{s})+\ell(P_{uv})+\ell(P_t) + k
    \ge d+k
    \,.
  \end{align*}
  We constructed a path of at least length~$d+k$ as required.
\end{proof}

\subsection{Proof of Lemma~\ref{lem: tw implies K4}: Large treewidth entails
  subdivided tetrahedra}
  \label{sec: tw-gives-tetrahedra}

To prove Lemma~\ref{lem: tw implies K4}, we require some preliminaries from
graph minors theory, among them a term for vertex sets that enjoy very favorable connectivity properties.
\begin{definition}[\cite{Diestel}]
 Let $G$ be a graph and $A,B\subseteq V(G)$. The pair $(A,B)$
is a \emph{separation} in $G$ if the sets~$A\setminus B$ and~$B\setminus A$
are non-empty and no edge runs between them. The \emph{order} of $(A,B)$
is the cardinality of $A\cap B$.

For $S\subseteq V(G)$, we say that $S$ is \emph{linked
}in $G$ if, for every $X,Y\subseteq S$ with $|X|=|Y|$, there are
$|X|$ vertex-disjoint paths between $X$ and $Y$ that intersect $S$ exactly at
its endpoints.
\end{definition}
The notion of left-containment conceptually connects separators and minor
models.
\begin{definition}
  Let $H$ be a graph on $k\in\mathbb{N}$ vertices.
  Recall Definition~\ref{def: minor-models} for the notion of a minor model.
  We say that $(A,B)$ \emph{left-contains} $H$ if $G[A]$ contains a minor model $f$ of $H$ with $|f(v)\cap(A\cap B)|=1$ for all $v\in V(H)$.
\end{definition}

With these definitions at hand, we can adapt a result by Leaf and Seymour~\cite{LeafS15}
to prove the following lemma on topological minor containment in graphs of sufficiently large treewidth.
For any forest $F$ on $k$ vertices, with maximum degree $3$, it asserts that graphs $G$ of treewidth
$\Omega(k)$ admit a separation such that one side contains $F$ as a topological minor,
with the branch vertices of this topological minor being contained in $A \cap B$ and linked in $G$.
We will use this lemma to complete the topological $F$-minor in~$G[A]$ to a larger graph by
using disjoint paths between vertices in $A\cap B$.

\begin{lemma}
\label{lem: linkage-topological} Let $F$ be a forest on $k>0$
vertices with maximum~degree $3$ and let $G$ be a graph. If $\tw(G)\geq\frac{3}{2}k-1$,
then $G$ has a separation $(A,B)$ of order $|V(F)|$
such that:
\begin{enumerate}
\item There is a topological minor model $(f,p)$ of $F$ in $G[A]$.
\item For every vertex $v\in V(F)$ of degree $\leq2$, we have $f(v)\in A\cap B$.
\item $A\cap B$ is linked in $G[B]$.
\end{enumerate}
\end{lemma}

We defer the proof of this lemma to \S\ref{sec:proof linkage-topological}.
Building upon Lemma~\ref{lem: linkage-topological}, we prove Lemma~\ref{lem: tw implies K4}
by adapting work of Raymond and Thilikos~\cite{RaymondT16}, who used
a variant of Lemma~\ref{lem: linkage-topological} to prove the existence of $k$-wheel minors in graphs
of treewidth $\Omega(k)$.
To this end, let $T$ and $P$ be obtained by $k$-subdividing the
full binary tree with $8$ leaves, and the path with $8$ vertices, respectively.
We invoke Lemma~\ref{lem: linkage-topological} with $F$ instantiated to the
disjoint union $T\cup P$. Since $F$ has $21 k + 2$ vertices, we obtain from Lemma~\ref{lem:
  linkage-topological}
that any graph $G$ with $\tw(G)\geq 32 k + 2\geq \frac{3}{2}\cdot (21 k +2 )$ has a
separation
$(A,B)$ of order $|V(F)|$ that contains~$F$ in $G[A]$ and has $A\cap B$
linked in~$G[B]$.

Let $X_{F}$ denote the eight leaves of $T$, and let $Y_{F}$
denote the eight non-subdivision vertices of $P$. Furthermore, let
$X_{G},Y_{G}\subseteq A\cap B$ denote the images of $X_{F}$ and
$Y_{F}$ in $G[A]$ under a topological minor model guaranteed by Lemma~\ref{lem: linkage-topological}.
Since $A\cap B$ is linked, we can find eight
disjoint paths connecting $X_{G}$ and $Y_{G}$ in~$G[B]$.
We then prove that, regardless of how these paths connect~$X_{G}$ and~$Y_{G}$,
they always complete the topological minor model of~$F$ to one of~$\subdivtetra k$ in $G$.
Lemma~\ref{lem: tw implies K4} then follows.

\begin{proof}[of Lemma~\ref{lem: tw implies K4}]
Let $k\in\mathbb{N}$ and let $G$ be a graph with $\tw(G)\geq 32 k + 2$.
As before, let~$T$ denote the full binary tree with $8$ leaves, with
root $r$, after each edge was subdivided $k$ times. Let
$P$ denote the path on $8$ vertices after subdividing each edge~$k$ times.

We write $X_{F}=\{x_{1},\ldots,x_{8}\}$ for the leaves of $T$, and
we write $Y_{F}=\{y_{1},\ldots,y_{8}\}$ for the vertices in $P$
that were not obtained as subdivision vertices. Finally, we write
$F$ for the disjoint union $T\cup P$ and consider $X_{F},Y_{F}\subseteq V(F)$.
Note that $|V(F)|= 21 k + 2$ and that the degree of all vertices in
$X_{F}\cup Y_{F}$ is bounded by $2$.

By Lemma~\ref{lem: linkage-topological}, there is a separation
$(A,B)$ in $G$ of order $|V(F)|$ such that $A\cap B$ is linked,
and there is a topological minor model $(f,p)$ of $F$ in $G[A]$
with $f(X_{F}\cup Y_{F})\subseteq A\cap B$. We write $X_{G}=\{f(v)\mid v\in X_{F}\}$
and $Y_{G}=\{f(v)\mid v\in Y_{F}\}$. In the following, we aim at
completing the subgraph induced by $(f,p)$ in $G$ to a $\subdivtetra k$
subgraph.

Since $A\cap B$ is linked in $G[B]$, there are vertex-disjoint paths
$L_{1},\ldots,L_{8}$ between~$X_{G}$ and~$Y_{G}$ in~$G[B]$ that
avoid $A\cap B$ except at their endpoints. For $i\in[8]$, denote
the endpoints of~$L_{i}$ in~$X_{G}$ and~$Y_{G}$ by $s_{i}$ and
$t_{i}$, respectively. Assume without limitation of generality (by
reordering paths) that $t_{i}=f(y_{i})$ holds for all $i\in[8]$.
Furthermore, for $x\in X_{G}$, write $\sigma(x)$ for the vertex
of $Y_{G}$ that $x$ is connected to via its path among $L_{1},\ldots,L_{8}$.

\global\long\def\root#1{\mathit{root}(#1)}
\global\long\def\lca#1#2{\mathit{lca}(#1,#2)}

\begin{figure}
\begin{center}
  \begin{tabular}{rrr}
  \includegraphics[scale=0.6]{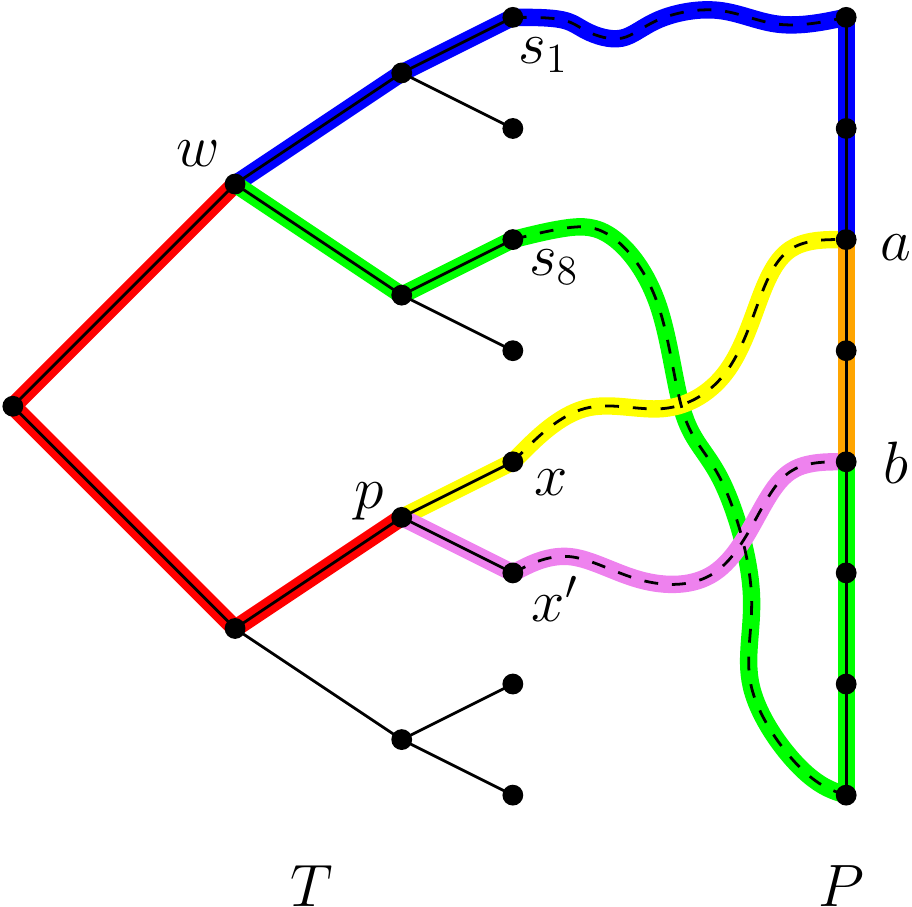}
  &
  ~~~~~
  &
  \includegraphics[scale=0.6]{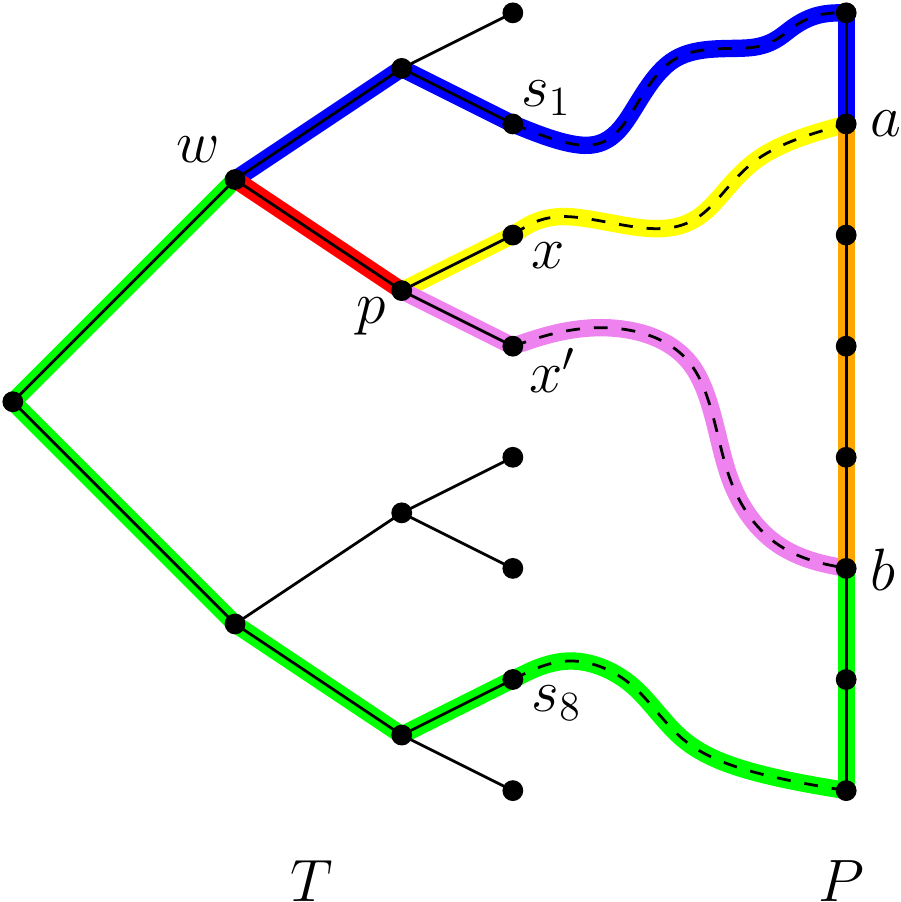}
  \\
  \\
  \includegraphics[scale=0.6]{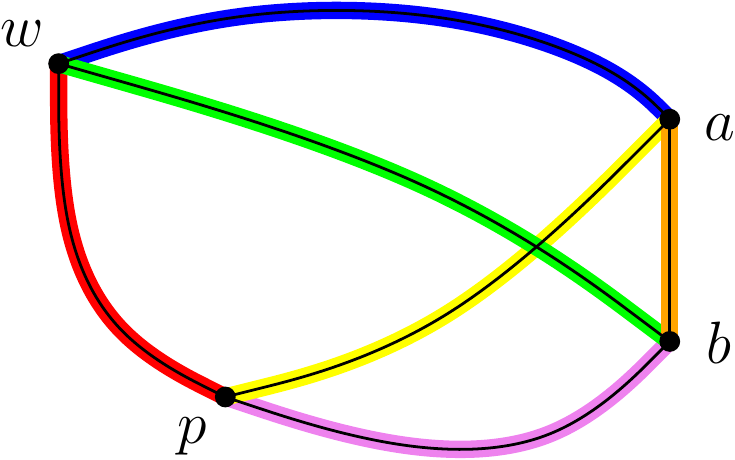}
  & &
  \includegraphics[scale=0.6]{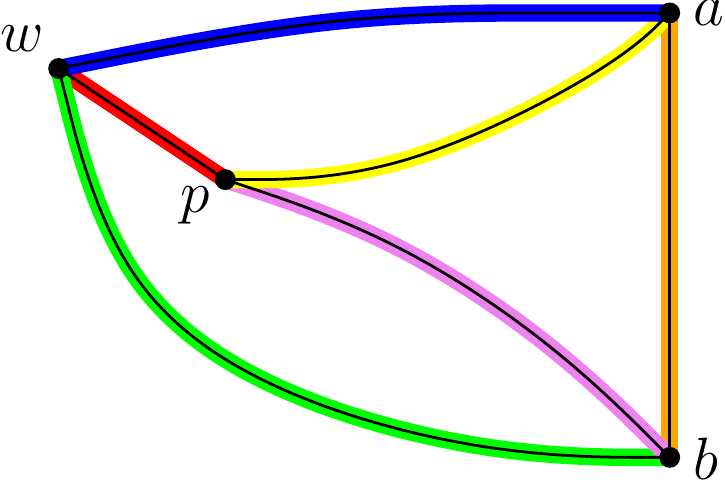}
  \end{tabular}
\end{center}

\caption{\label{fig: minor-cases}%
  The two cases relevant in the proof of Lemma~\ref{lem: tw implies K4}. The top figures depict the cases and the bottom figures show the corresponding $\subdivtetra k$ subgraphs.
  The \emph{bullets} correspond to the main vertices of the tree~$T$ and the
  path~$P$.
  The \emph{lines} represent a path of length~$k$, created by subdividing the
  original tree and path.
  The \emph{dashed curves} correspond to the linkage that is guaranteed to exist
  between the 8 leaves of~$T$ and the 8 main vertices of~$P$.
  The vertices $w$, $p$, $a$, and $b$ are the branch vertices of the
  $\subdivtetra k$-graph that we find.
  The vertex $s_1$ is the linkage-partner of the first path vertex, and $s_8$ is
  the linkage-partner of the last path vertex.
  \emph{Top left:} In Case~1, $s_1$ and $s_8$ have a least common ancestor~$w$ that
  is not the root.
  \emph{Top right:} In Case~2, $s_1$ and $s_8$ have the root as their least common
  ancestor.
  \emph{Bottom:} A schematic view of the corresponding $\subdivtetra k$ subgraphs, where each of the six subdivided edges (that is, paths of length greater than $k$) is shown in a different color. The same color is used to highlight the path in the graph above.
}
\end{figure}

Let $S$ denote the image of $T$ under $(f,p)$, which
is a tree; let $\root S=f(r)$. Write~$S_{1}, S_{2}$
for the two subtrees of~$S$ rooted at the children of~$\root S$.
Let $\lca{s_{1}}{s_{8}}$ denote the lowest common ancestor
of $s_{1}$ and $s_{8}$ in $S$. We distinguish two cases (see Figure~\ref{fig: minor-cases}).
\begin{description}
\item [{Case~1:}] We have $\lca{s_{1}}{s_{8}}\neq\root S$. That is, $s_{1}$
and $s_{8}$ are both in $S_{1}$ or both in $S_{2}$. Assume without
limitation of generality that $s_{1},s_{8}\in V(S_{1})$, as the argument
proceeds symmetrically otherwise. Let $x$ and $x'$ be two distinct leaves of
$S_{2}$. Then we find a $\subdivtetra k$ in $G$ by defining
branch vertices $w=\lca{s_{1}}{s_{8}}$, $p=\lca x{x'}$, $a=\sigma(x)$,
and $b=\sigma(x')$.
Note that $p\not\in\set{x,x'}$ and that the four vertices are distinct.

We realize the edge $pw$ along the $(p,w)$-path present in $S$,
and $ab$ along the $(a,b)$-path present in $P$.
We realize $pa$ by concatenating the $(p,x)$-path in $S$ and the $(x,a)$-path
in $G[B]$, and we realize $pb$ likewise.
To realize $wa$, we proceed as follows: If $a$ precedes $b$ in the order on $P$,
then concatenate the $(w,s_{1})$-path in $S$ with $L_{1}$ and the
$(y_{1},a)$-path in $P$. If $b$ precedes $a$, then concatenate
the $(w,s_{8})$-path in $S$ with $L_{8}$ and the $(y_{8},a)$-path
in $P$. Realize $wb$ symmetrically.
Then every edge between pairs in $\{w,p,a,b\}$ is realized,
and it is so by a path of length at least $k$. This gives a topological
minor model of $\subdivtetra k$ in $G$.

\item [{Case~2:}] We have $\lca{s_{1}}{s_{8}}=\root S$. That is, $s_{1}$
and $s_{8}$ are in different subtrees $S_{1}$ and $S_{2}$.
Let~$R$ be a subtree of height $2$ in $S$ that is disjoint from the
$(s_{1},s_{8})$-path in $S$. It is easy to verify that such a subtree
indeed exists; denote its root by $p$, its leaves by $x,x'$, and
its parent in $S$ by $w$. Furthermore, define $a=\sigma(x)$ and
$b=\sigma(x')$. We declare $\{w,p,a,b\}$ as branch vertices and
connect them as in the previous case. 
\end{description}
In both cases, the constructed topological minor model shows that~$G$ contains a
$\subdivtetra k$ subgraph. This proves the lemma.
\end{proof}

\subsection{Proof of Lemma~\ref{lem: linkage-topological}}\label{sec:proof linkage-topological}

For the following part, we need to define the notion of a \emph{minor model}.
Note that only \emph{topological} minor models were defined in the main text.
\begin{definition}\label{def: minor-models-apdx}
Let $H$ and $G$ be undirected.
A \emph{minor model} of $H$ in $G$ is a function $f:V(H)\to2^{V(G)}$ such that
\begin{enumerate}
\item $G[f(v)]$ is connected for all $v\in V(H)$, and
\item $f(u)\cap f(v)=\emptyset$ for all $u,v\in V(H)$ with $u\neq v$,
and
\item for all $uv\in E(H)$, there is an edge in $G$ from a vertex in $f(u)$
to a vertex in $f(v)$.
\end{enumerate}
\end{definition}

Furthermore, we require the notion of left-containment:

\begin{definition}
If $G$ is a graph with separation $(A,B)$, we say that $(A,B)$ \emph{left-contains}
$H$ if $G[A]$ contains a minor model $f$ of $H$ with $|f(v)\cap(A\cap B)|=1$
\end{definition}

We can now state a lemma by Leaf and Seymour~\cite{LeafS15} that can be easily adapted
to obtain Lemma~\ref{lem: linkage-topological}.

\begin{lemma}[\cite{LeafS15}]
\label{lem: linkages}Let $F$ be a forest on $k>0$ vertices and
let $G$ be a graph.
If $\tw(G)\geq\frac{3}{2}k-1$, then there exists
a separation $(A,B)$ of order~$|V(F)|$ in $G$ such that $(A,B)$
left-contains~$F$, and $A\cap B$ is linked in $G[B]$.
\end{lemma}

Finally, we prove Lemma~\ref{lem: linkage-topological}.

\begin{proof}[of Lemma~\ref{lem: linkage-topological}]
  Let $(A,B)$ be the separation from Lemma~\ref{lem: linkages}. Then, the third
  condition holds so we just need to prove the first two conditions. By the same
  lemma, $G[A]$ contains a minor model $f'$ of $F$ with $|f'(v)\cap(A\cap B)|=1$
  for all $v\in V(F)$. It will be convenient to fix a spanning tree inside each
  $G[f'(v)]$ for every $v\in V(F)$; let us denote it by $T(v)$. Additionally,
  for every $v_1v_2\in E(F)$, we will fix one edge $u_1u_2\in E(G)$ such that
  $u_i\in f'(v_i)$. We define the topological minor model $(f,p)$ as follows.
  For every $v\in V(F)$ of degree $\leq2$, let $f(v)=u$ where $u\in
  f'(v)\cap(A\cap B)$, the only such vertex. This satisfies the second condition
  of the corollary.

  For $v\in V(F)$ of degree $3$, let $v_1,v_2,v_3$ be its adjacent vertices in
  $F$. Then, for each edge $vv_i$ let $u_iu'_i$ be the corresponding fixed edge
  where $u_i\in f'(v)$ (and $u'_i \in f'(v_i)$). If $u_1,u_2,u_3$ are not
  distinct, then choose one of them to be $f(v)$. If they are distinct, then
  take the spanning tree $T(v)$, root it at $u_3$ and let $f(v)$ be the vertex
  that is the lowest common ancestor of $u_1$ and $u_2$.

  We defined $f$, now we define the paths $p$. For edge $v_1v_2\in E(F)$, let
  $p(v_1v_2)$ be the path defined by following the spanning tree $T(v_1)$ from $f(v_1)$
  to the edge $u_1u_2\in E(G)$ that we fixed for $v_1v_2\in E(F)$, and then
  following the spanning tree $T(v_2)$ to $f(v_2)$. It remains to show that the
  paths are vertex-disjoint, except for their endpoints. Every path uses exactly
  one edge not in $G[f'(v)]$ for some $v\in V(F)$; these edges are distinct as
  they correspond to different edges of $F$. Therefore, path intersections could
  only happen inside one of the $G[f'(v)]$ graphs. However, the selected vertex
  $f(v)$ is connected to the vertices $u_1,u_2,u_3$ defined in the previous
  paragraph via disjoint paths in the spanning tree, proving that the paths $p$
  are vertex-disjoint within $G[f'(v)]$ apart from the endpoints. This proves
  that $(f,p)$ is a topological minor model and concludes the proof of the
  corollary.
\end{proof}

\section{Dynamic programming algorithm for exact detour}\label{sec:DP}
We devise an algorithm for \probEXDet using a reduction to \probEXKPath, the
problem that is given $(G,s,t,k)$ to determine whether there is an~$(s,t)$-path
of length exactly~$k$.
\begin{theorem}\label{thm: exact detour algorithm}%
  \probEXDet is fixed-parameter tractable.
  In particular, it has a bounded-error randomized algorithm with running
  time~$2.746^{k}\poly(n)$, and a deterministic algorithm with running time
  $6.745^{k} \poly(n)$.
\end{theorem}
\begin{proof}
  Let $(G,s,t,k)$ be an instance of \probEXDet.
  We use Lemma~\ref{lem: DP reduction} and run the deterministic polynomial-time
  reduction in algorithm~A on this instance, which makes queries to
  $\probEXKPath$ whose parameter $k'$ is at most $2k+1$.
  To answer these queries, we use the best known algorithm as a subroutine.
  Using the deterministic algorithm by Zehavi~\cite{Zehavi14}, we obtain a
  running time of $2.597^{k'} \cdot \poly(n) \le 6.745^k\cdot\poly(n)$ for
  \probEXDet.
  Using the randomized algorithm by Björklund et al.~\cite{Bjorklund2017119}, we obtain
  a running time of $1.657^{k'}\cdot\poly(n)\le 2.746^{k}\cdot\poly(n)$.
\end{proof}

\begin{figure}
\begin{center}
  \includegraphics[scale=.6]{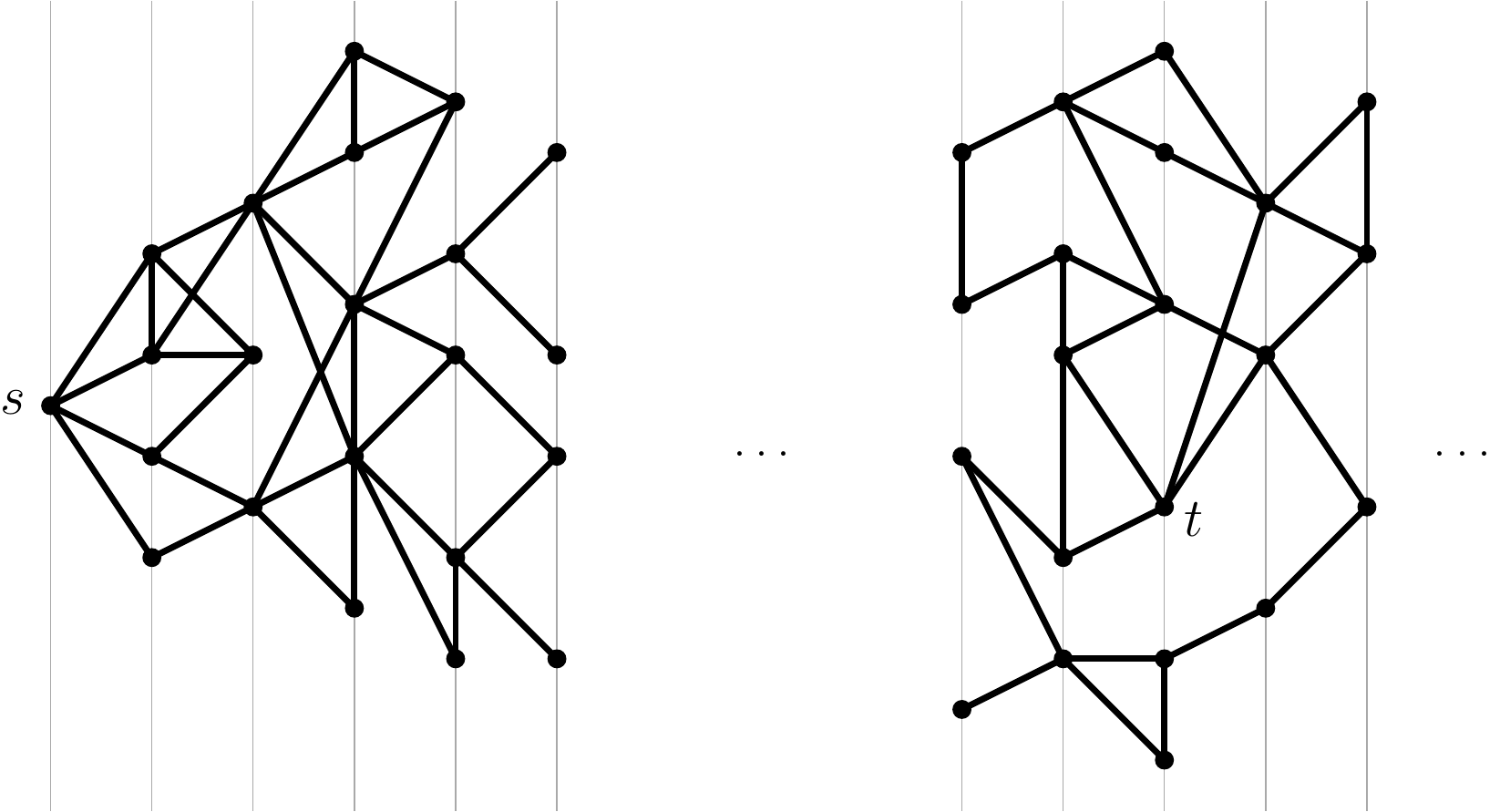}\\
\end{center}
\caption{\label{fig: BFS_layers}%
  The \emph{fine vertical lines} in this drawing of an example graph represent
  distance layers, that is, vertices whose distance~$d(v)$ from $s$ is equal.
  }
\end{figure}
\begin{figure}
\begin{center}
  \includegraphics[scale=.7]{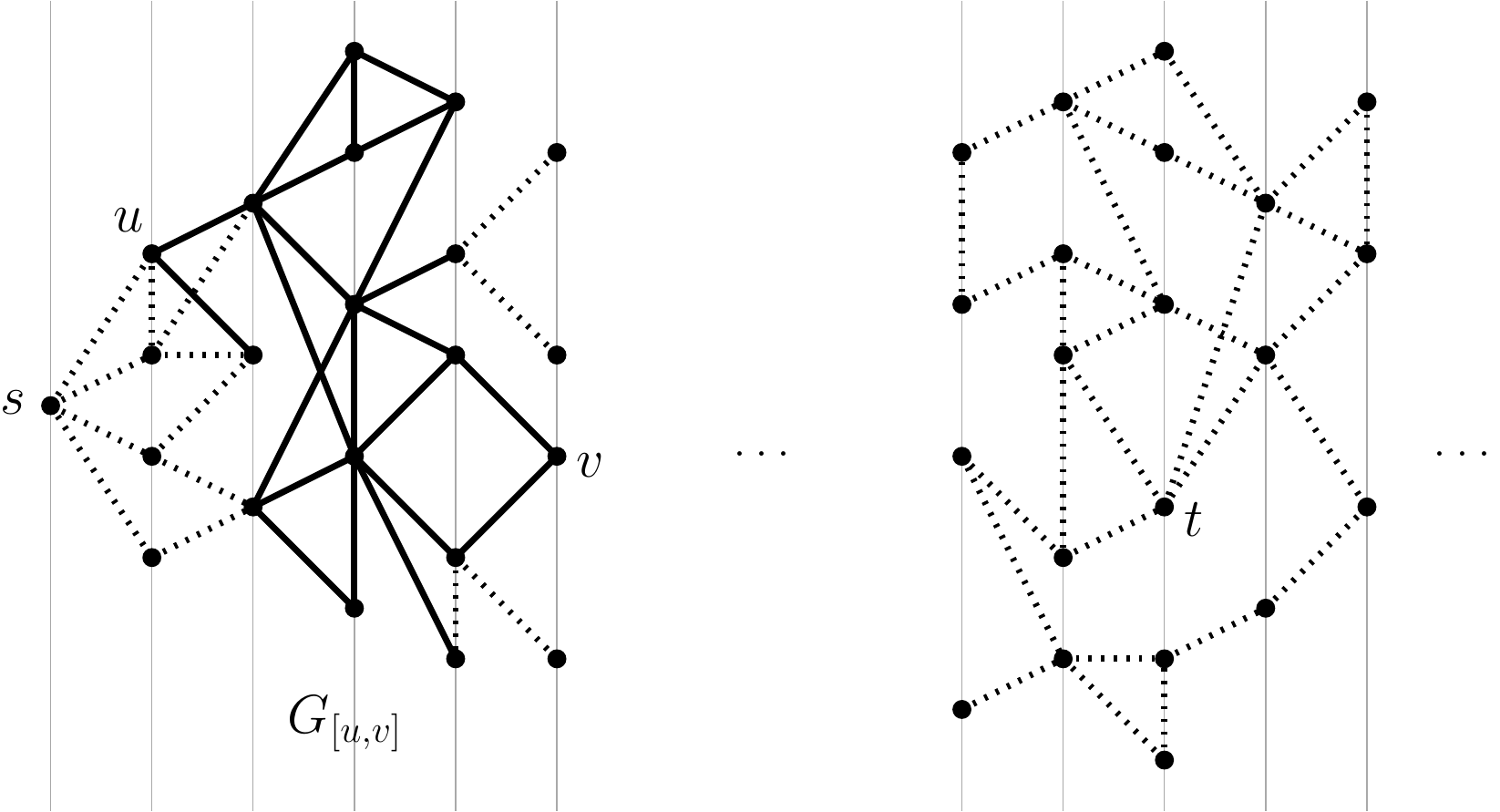}
\end{center}
\caption{\label{fig: graph Guv}%
  The \emph{solid} edges are the edges of an example graph $G_{[u,v]}$; the
  other edges of~$G$ are \emph{dashed}.
  }
\end{figure}

Before we state the algorithm, let us introduce some notation.
Let $s,t\in V(G)$.
For any $x\in V(G)$, we abbreviate $\dist_G(s,x)$, that is, the distance from
$s$ to $x$ in $G$, with~$d(x)$, and we let the $i$-th layer of~$G$ be the set of
vertices~$x$ with~$d(x)=i$ (see Figure~\ref{fig: BFS_layers}).
For $u,v\in V(G)$ with $d(u)<d(v)$, we write~$G_{[u,v]}$ for the graph~$G[X]$
induced by the vertex set $X$ that contains $u$, $v$, and all vertices $x$ with
$d(u)<d(x)<d(v)$ (see Figure~\ref{fig: graph Guv}).
We also write $G_{[u,\infty)}$ for the graph $G[X]$ induced by the vertex set
$X$ that contains $u$ and all vertices~$x$ with $d(u)<d(x)$.
These graphs can be computed in linear time using breadth-first search
starting at $s$.
We now describe an algorithm for \probEXDet that makes queries to an oracle
for \probEXKPath.

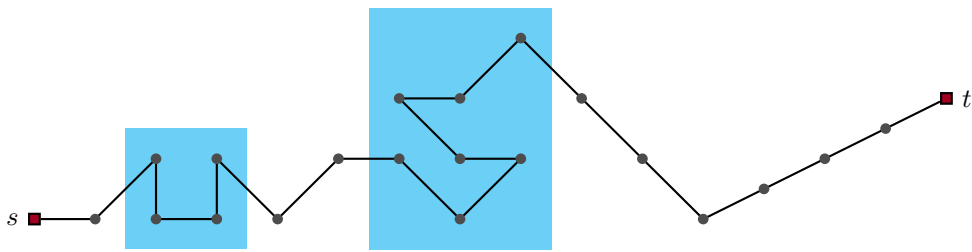
\begin{figure}[tp]
\begin{center}
  \begin{tikzpicture}[scale=.8]
    \node[uv,label=left:{$s$}] (s) at (0,0) {};
    \node[uv,label=right:{$t$}] (t) at (15,2) {};

    \path[fill=cyan!50] (1.5,1.5) rectangle (3.5,-.5);
    \path[fill=cyan!50] (5.5,3.5) rectangle (8.5,-.5);

    \draw[thick] (s)
    -- (1,0)  node[vtx] {}
    -- (2,1)  node[vtx] {}
    -- (2,0)  node[vtx] {}
    -- (3,0)  node[vtx] {}
    -- (3,1)  node[vtx] {}
    -- (4,0)  node[vtx] {}
    -- (5,1)  node[vtx] {}
    -- (6,1)  node[vtx] {}
    -- (7,0)  node[vtx] {}
    -- (8,1)  node[vtx] {}
    -- (7,1)  node[vtx] {}
    -- (6,2)  node[vtx] {}
    -- (7,2)  node[vtx] {}
    -- (8,3)  node[vtx] {}
    -- (9,2)  node[vtx] {}
    -- (10,1)  node[vtx] {}
    -- (11,0) node[vtx] {}
    -- (12,.5) node[vtx] {}
    -- (13,1) node[vtx] {}
    -- (14,1.5) node[vtx] {}
    -- (t);
  \end{tikzpicture}
\end{center}
\caption{\label{fig: detour setU}%
  This is an example of a long $(s,t)$-path in a graph with~$d_G(s,t)=15$; the
  distance from~$s$ increases from left to right as in Figure~\ref{fig:
    BFS_layers}.
  The path has length~$21$, so it has~$k=6$ more edges than a shortest path.
  Each of the five marked layers (\emph{cyan shading}) contains more than one
  vertex of the path, and any path of length~$21$ can have at most~$6$ such
  layers.
  }
\end{figure}

The general idea is as follows.
Let $G$ be an undirected graph, and consider an $(s,t)$-path~$P$ of length~$d+k$
where $d=d_G(s,t)$, and let~$x$ be a token that travels along this path from~$s$
to~$t$.
As the token advances one step in the path, the number~$d(x)$ can be
incremented, decremented, or stay the same.
When~$x$ moves from~$s$ to~$t$, we must increment~$d(x)$ at least~$d$ times, can
decrement it at most $k/2$ times, and keep it unchanged at most~$k$ times; the
reason is that the path must reach $t$ but must use exactly~$k$ edges more than
a shortest path.
The crucial observation is that there are at most~$k$ different layers whose
intersection with the path~$P$ contains more than one vertex (see Figure~\ref{fig: detour setU}).
The idea for the algorithm is to guess the layers with more than one vertex and
run an algorithm for \probEXKPath on them.

\begin{algor}{A}{\probEXDet}{Given $(G,s,t,k)$, this algorithm decides whether
    the graph~$G$ contains an $(s,t)$-path of length exactly $\dist_G(s,t)+k$.}
\item[A1] (Initialize table) For each $x\in V(G)$ with $d(x)\le d(t)$,
  set $T[x] = \emptyset$.

  \textit{When the algorithm halts, every entry~$T[x]$ of the table is meant to
    satisfy the following property $Q_x$:
    For each integer $\ell$ with $d(t)-d(x)\le\ell\le d(t)-d(x)+k$, the set
    $T[x]$ contains~$\ell$ if and only if $G_{[x,\infty)}$ contains an
    $(x,t)$-path of length~$\ell$.}

\item[A2] (Compute entries for the last $k+1$ layers)
  For each $x\in V(G)$ with $d(t)-k\le d(x)\le d(t)$, let $T[x]$ be the set of all
  integers $\ell$ with $\ell\in\set{0,\dots,2k}$ such that there is an
  $(x,t)$-path of
  length~$\ell$ in~$G_{[x,\infty)}$ (that is, call \probEXKPath$(G_{[x,\infty)},x,t,\ell)$).

  \textit{When this step finishes, all vertices~$x$ in the last $k+1$ layers
    satisfy property $Q_x$.}

\item[A3] (Inductively fill in earlier layers)
  For each $d$ from $d(t)-k-1$ down to $0$, for each $x$ with $d(x)=d$, and for
  each $y$ with $d(x)<d(y)\le d(x)+k+1$, we do the following:
  \begin{description}
    \item[A3a]
      Compute the set $L$ of all $\ell'\in\set{0,\dots,2k+1}$ such that there
      is an $(x,y)$-path of length~$\ell'$ in~$G_{[x,y]}$  (that is, call
      \probEXKPath$(G_{[x,y]},x,y,\ell')$).
    \item[A3b]
      Set $T[x] := T[x] \cup (L+T[y])$.
  \end{description}

  \textit{We will show that, when all vertices of a layer~$d$ have been
    considered, all vertices~$x$ in the layers~$d$ and higher satisfy
    property~$Q_x$.}

\item[A4] Accept if and only if $(\dist_G(s,t)+k)\in T[s]$ holds.
\end{algor}
\begin{lemma}\label{lem: DP reduction}
  Algorithm A is a polynomial-time Turing reduction from \probEXDet to
  \probEXKPath; on instances with parameter~$k$, all queries have
  parameter at most $2k+1$.
\end{lemma}
\begin{proof}
  The running time of A is polynomially bounded since breadth-first search can
  be used to discover all partial graphs $G_{[x,y]}$ and $G_{[x,\infty)}$, and
  we loop at most over every pair of vertices in A2 and A3.
  For the parameter bound, note that the queries in A2 and A3 are for paths of
  length at most $2k$ and $2k+1$, respectively.
  It remains to prove the correctness.

  We execute algorithm A on an instance $(G,s,t,k)$.
  For the correctness, it suffices to prove that property~$Q_s$ holds at the end
  of the execution:
  Note that $\ell$ with $\ell=d_G(s,t)+k$ lies in the interval
  $[d(t)-d(s),d(t)-d(s)+k]$ since $d(s)=0$ and $d(t)=d_G(s,t)$ holds.
  Moreover, we have $G_{[s,\infty)}=G$.
  Thus $Q_s$ guarantees that $\ell\in T[s]$ holds if and only if $G$ contains an
  $(s,t)$-path of length~$\ell$,
  which by step A4 implies that A accepts if and only if $(G,s,t,k)$ is a
  yes-instance of \probEXDet.
  Therefore it remains to prove that~$Q_s$ holds at the end of the execution of~A.
  We do so using the following claim.

  \textit{Claim:}
  For all $x$ with $0\le d(x)\le d(t)$, property~$Q_x$ holds forever after the
  entry $T[x]$ is written to for the last time.

  We prove this claim by induction on $d(x)$.
  For the base case, let $x$ be a vertex with $d(x)\geq d(t)-k$.
  The entry $T[x]$ is only written to in step A2.
  To prove that $Q_x$ holds after~A2, let $\ell$ be an integer with
  $d(t)-d(x)\le \ell \le d(t)-d(x)+k$.
  Note that $d(t)-d(x)\ge 0$ and $d(t)-d(x)+k \le d(t)-\paren*{d(t)-k}+k\le 2k$
  holds, and so step A2 adds $\ell$ to $T[x]$ if and only if the
  graph~$G_{[x,\infty)}$ contains an $(x,t)$-path of length~$\ell$.
  Therefore,~$Q_x$ holds forever after~A2 has been executed.

  For the induction step, let $x$ be a vertex with $d(x)<d(t)-k$.
  By the induction hypothesis,~$Q_y$ holds for all~$y$ with $d(y)>d(x)$.
  The entry $T[x]$ is only written to in step~A3b, and when it is first written
  to, the outer $d$-loop in A3 has fully processed all layers larger
  than~$d(x)$.
  Thus already when~$T[x]$ is written to for the first time, $Q_y$ holds for
  all~$y$ with $d(y)>d(x)$.
  Let $T$ be the table right after A3b writes to $T[x]$ for the last time.
  It remains to prove that~$T[x]$ satisfies $Q_x$.
  Let $\ell$ be an integer with $d(t)-d(x)\le \ell\le d(t)-d(x)+k$.

  \textit{Claim:} There is an $(x,t)$-path of length~$\ell$ if and only if
  $T[x]$ contains~$\ell$.

  For the forward direction, let $P$ be an $(x,t)$-path in $G_{[x,\infty)}$ of
  length exactly~$\ell$.
  There are exactly~$\ell$ vertices~${u\in V(P)\setminus\set{x}}$.
  Moreover, since every edge $uv\in E(P)$ satisfies $\abs{d(u)-d(v)}\le 1$,
  every $d\in\set{d(x)+1,\dots,d(t)}$ must have some vertex $u\in V(P)$ with
  $d(u)=d$.
  Since $\ell\le d(t)-d(x) + k$, there are at most $k$ distinct~$d$ where more
  than one vertex~$u\in V(P)$ satisfies $d(u)=d$.
  By the pigeon hole principle, there exists an integer~$d$ in the
  $(k+1)$-element set $\set{d(x)+1,\dots,d(x)+k+1}$ such that there is exactly
  one vertex $y\in V(P)$ with $d(y)=d$.

  Let $P_{[x,y]}$ be the subpath of $P$ between~$x$ and~$y$, and let $\ell'$ be
  its length.
  By construction, $P_{[x,y]}$ is an $(x,y)$-path in~$G_{[x,y]}$.
  Moreover, we have $\ell' \le \ell - \paren*{d(t)-d(y)}$
  since~$V(P)\setminus\set{x}$ contains $\ell$ vertices~$u$, at least
  $d(t)-d(y)$ of which satisfy~$d(u)>d(y)$.
  By choice of~$\ell$ and~$y$, we obtain~$\ell'\le d(y)-d(x)+k\le 2k+1$.
  For this setting of $y$ and $\ell'$, step A3a detects the path $P_{[x,y]}$
  and~$\ell'$ is added to the set~$L$.
  The second piece $P_{[y,\infty]}$ of the path $P$ is a $(y,t)$-path in
  $G_{[y,\infty]}$ of some length~$\ell''$ between $d(t)-d(y)$ and
  $d(t)-d(y)+k$; since $Q_y$ holds when~A3b is executed for $x$ and $y$, the
  set~$T[y]$ contains~$\ell''$, and so $\ell=\ell'+\ell''$ gets added to $T[x]$.
  Since elements never get removed from~$T[x]$, the forward direction of the
  claim holds.

  For the backward direction of the claim, assume that $T[x]$ contains~$\ell$.
  This means that~$\ell$ is added in step~A3b during the execution of the
  algorithm; in particular, consider the variables $y\in V(G)$, $\ell'\in L$,
  and $\ell''\in T[y]$ when $\ell=\ell'+\ell''$ is added to $T[x]$.
  By the induction hypothesis, $\ell''\in T[y]$ implies that there is a
  $(y,t)$-path in~$G_{[y,\infty)}$ of length~$\ell''$.
  Moreover, $\ell'$ was set in A3a in such a way that there is an $(x,y)$-path
  of length~$\ell'$ in the graph~$G_{[x,y]}$.
  Combined, these two paths yield a single $(x,t)$-path in~$G_{[x,\infty)}$ of
  length~$\ell$.
  The backward direction of the claim follows.
\end{proof}

The randomized algorithm of Björklund et al.~\cite{Bjorklund2017119} is for a variant of
\probEXKPath where the terminal vertices~$s$ and~$t$ are not given, that is, any
path of length exactly~$k$ yields a \texttt{YES}-instance.
Their algorithm applies to our problem as well, with the same running time.
We sketch an argument for this observation here.
Recall that the idea is to reduce the problem to checking whether a certain
polynomial is identically zero -- this polynomial is defined by summing over all
possible labelled walks of length~$k$ (see
\cite[Sec.~10.4.3]{cygan2015parameterized}).
We modify the polynomial by adding two leaf-edges, one incident to~$s$ and one
to~$t$, and restricting our attention only to $(k+2)$-walks that contain these
two edges.
The required information for such walks can still be computed efficiently as
before.
The crux of the proof is that walks that are not paths cancel out over a field
of characteristic two; this argument works by a local re-orientation of segments
of the walk -- an operation that does not change the vertices of the walk and
must therefore keep~$s$ and~$t$ fixed.
The graph $G$ contains a $k$-path if and only if the polynomial is not
identically zero; this property remains true in our case.
The rest of the argument goes through as before, so the algorithm of Björklund et al.~applies to~\probEXKPath with no significant loss in the running time.

The deterministic algorithm of Zehavi~\cite{Zehavi14} also does not expect the
terminal vertices to be given, but this algorithm works for the weighted version
of the problem.
In the weighted $k$-path problem, we are given a graph~$G$, weights~$w_e$ on
each edge, a number~$k$, and a number~$W$, and the question is whether there is
a path of length exactly~$k$ such that the sum of all edge weights along the
path is at most~$W$.
We observe the following simple reduction from \probEXKPath (with terminal
vertices~$s$ and~$t$) to the weighted $k$-path problem (without terminal
vertices):
Every edge gets assigned the same edge weight~$2$, except for the new leaf-edges
at~$s$ and~$t$, which get edge weight~$1$.
Now every path with exactly~$k+2$ edges has weight at most~$W=2k+2$ if and only
if the first and the last edges of the path are the leaf-edges we added.
Due to this reduction, Zehavi's algorithm applies to \probEXKPath with no
significant loss in the running time.

Theorem~\ref{thm: exact detour algorithm} follows from Algorithm A by using
either the algorithm of Björklund et al.~\cite{Bjorklund2017119} or
Zehavi~\cite{Zehavi14} as the oracle.
We remark that Theorem~\ref{thm: exact detour algorithm} and Algorithm~A apply
to directed graphs as well, in which case an algorithm for \probEXKPath in
directed graphs needs to be used (color coding yields the fastest randomized
algorithm~\cite{AlonYZ}, while Zehavi's deterministic algorithm also applies to
directed graphs).

\section{Search-to-decision reduction}\label{sec:searchtodecision}

Our graph-minor based algorithm for \probLongDet does not directly construct a good
path since the algorithm merely says ``yes'' when the tree width is large
enough.
Similarly, our dynamic programming algorithm uses an algorithm for \probEXKPath
as a subroutine, and these algorithms also do not typically find a path
directly.

In this section, we present a search-to-decision reduction for \probLongDet and
\probEXDet that uses a simple downward self-reducibility argument.
In the interest of brevity, we focus on \probLongDet:
Given a decision oracle for this problem, we show how to construct a detour with
only polynomial overhead in the running time.

\begin{algor}{B}{Search-to-decision reduction}{Given $(G,s,t,k)$ and access to
    an oracle for \probLongDet, this algorithm computes an $(s,t)$-path of
    length at
    least $\dist_G(s,t)+k$.}
\item[B0] (Trivial case)
  If $(G,s,t,k)$ is a no-instance of \probLongDet, halt and reject.
\item[B1] (Add a new shortest path)
  Add $d:=\dist_G(s,t)$ new edges to $G$, forming a new shortest $(s,t)$-path $p_1,\dots,p_d$.
\item[B2] (Delete unused edges)
  For each $e\in E(G)\setminus \set{p_1,\dots,p_d}$:
      If $(G-e,s,t,k)$ is a yes-instance of \probLongDet, then set $G:= G-e$.
\item[B3] (Delete the added path)
  Let $G:=G-\set{p_1,\dots,p_d}$.
\item[B4] (Output detour)
  Now $G$ is an $(s,t)$-path of length at least $\dist_G(s,t)+k$.
\end{algor}
\begin{lemma}
  Algorithm B is a polynomial-time algorithm when given oracle access to
  \probLongDet,
  and it outputs a path of length at least~$\dist_G(s,t)+k$.
\end{lemma}
\begin{proof}
  It is clear that B runs in polynomial time; we only need to show correctness.
  Let $G_0$ be the graph at the beginning of the algorithm, let $G_1$ be the
  remaining graph after B2.
  If $(G_0,s,t,k)$ is a yes-instance, then $(G_1,s,t,k)$ is also a yes-instance.
  Moreover, deleting any edge from $E(G)\setminus\set{p_1,\dots,p_d}$ would turn
  it into a no-instance.
  Since $(G_1,s,t,k)$ is a yes-instance, it contains an $(s,t)$-path
  $q_1,\dots,q_\ell$ for $\ell\ge \dist_G(s,t)+k$.
  Since the yes-instance is minimal, we have
  $E(G_1)=\set{p_1,\dots,p_d}\cup\set{q_1,\dots,q_\ell}$.

  Finally, since $p_1,\dots,p_d$ got added to $G$ as a new path, it is
  edge-disjoint from every other $(s,t)$-path in~$G$. Therefore, by removing
  $\set{p_1,\dots,p_d}$ from $G_1$, we get the path $q_1,\dots,q_\ell$, of
  length at least $\dist_G(s,t)+k$.
\end{proof}

\section{Conclusion}

We conclude with the following open problem: what is the complexity of
\probLongDet in directed graphs?  So far, our attempts  to mimic the algorithm
for undirected graphs did not work.  By the celebrated work of Kawarabayashi and
Kreutzer~\cite{KawarabayashiK15},  every directed graph of sufficiently large
directed treewidth contains a large directed grid as a butterfly minor. It is
tempting to use this theorem in order to obtain a Win/Win algorithm for
\probLongDet on directed graphs; however, there are several obstacles on this
path.
Actually we do not know if  the problem is in the class XP, that is, if there is
an algorithm that solves directed \probLongDet in time $n^{f(k)}$ for some
function~$f$.
Can one even find an $(s,t)$-path of length $\ge d_G(s,t)+1$
in polynomial time?

For undirected planar graphs, by standard bidimensionality
arguments \cite{DemaineFHT05jacm}, our algorithm can be sped up to run in time
$2^{\Oh(\sqrt{k})} n^{\Oh(1)}$, but we do not know if \probLongDet in directed
planar graphs is in XP.

\subparagraph*{Acknowledgments.}
We thank
Daniel Lokshtanov,
Meirav Zehavi,
Petr Golovach,
Saket Saurabh,
Stephan Kreutzer, and
Tobias Mömke
for helpful discussions and answers.

\bibliographystyle{plainurl}
\bibliography{detours}

\begin{thebibliography}{10}

\bibitem{DBLP:journals/algorithmica/AlonGKSY11}
Noga Alon, Gregory Gutin, Eun~Jung Kim, Stefan Szeider, and Anders Yeo.
\newblock Solving {MAX-$r$-SAT} above a tight lower bound.
\newblock {\em Algorithmica}, 61(3):638--655, 2011.
\newblock \href {http://dx.doi.org/10.1007/s00453-010-9428-7}
  {\path{doi:10.1007/s00453-010-9428-7}}.

\bibitem{AlonYZ}
Noga Alon, Raphael Yuster, and Uri Zwick.
\newblock Color-coding.
\newblock {\em Journal of the ACM}, 42(4):844--856, 1995.
\newblock \href {http://dx.doi.org/10.1145/210332.210337}
  {\path{doi:10.1145/210332.210337}}.

\bibitem{Bjorklund2017119}
Andreas Björklund, Thore Husfeldt, Petteri Kaski, and Mikko Koivisto.
\newblock Narrow sieves for parameterized paths and packings.
\newblock {\em Journal of Computer and System Sciences}, 87:119--139, 2017.
\newblock \href {http://dx.doi.org/10.1016/j.jcss.2017.03.003}
  {\path{doi:10.1016/j.jcss.2017.03.003}}.

\bibitem{Bodlaender93a}
Hans~L. Bodlaender.
\newblock On linear time minor tests with depth-first search.
\newblock {\em Journal of Algorithms}, 14(1):1--23, 1993.
\newblock \href {http://dx.doi.org/10.1006/jagm.1993.1001}
  {\path{doi:10.1006/jagm.1993.1001}}.

\bibitem{rank-treewidth}
Hans~L. Bodlaender, Marek Cygan, Stefan Kratsch, and Jesper Nederlof.
\newblock Deterministic single exponential time algorithms for connectivity
  problems parameterized by treewidth.
\newblock {\em Information and Computation}, 243:86--111, 2015.
\newblock \href {http://dx.doi.org/10.1016/j.ic.2014.12.008}
  {\path{doi:10.1016/j.ic.2014.12.008}}.

\bibitem{BodlaenderDDFLP16}
Hans~L. Bodlaender, P{\aa}l~Gr{\o}n{\aa}s Drange, Markus~S. Dregi, Fedor~V.
  Fomin, Daniel Lokshtanov, and Michal Pilipczuk.
\newblock A $c^k n$ $5$-approximation algorithm for treewidth.
\newblock {\em SIAM Journal on Computing}, 45(2):317--378, 2016.
\newblock \href {http://dx.doi.org/10.1137/130947374}
  {\path{doi:10.1137/130947374}}.

\bibitem{ChekuriC13}
Chandra Chekuri and Julia Chuzhoy.
\newblock Polynomial bounds for the grid-minor theorem.
\newblock {\em Journal of the ACM}, 63(5):40:1--40:65, 2016.
\newblock \href {http://dx.doi.org/10.1145/2820609}
  {\path{doi:10.1145/2820609}}.

\bibitem{ChenKLMR09}
Jianer Chen, Joachim Kneis, Songjian Lu, Daniel M{\"o}lle, Stefan Richter,
  Peter Rossmanith, Sing-Hoi Sze, and Fenghui Zhang.
\newblock Randomized divide-and-conquer: Improved path, matching, and packing
  algorithms.
\newblock {\em SIAM Journal on Computing}, 38(6):2526--2547, 2009.
\newblock \href {http://dx.doi.org/10.1137/080716475}
  {\path{doi:10.1137/080716475}}.

\bibitem{ChenLSZ07}
Jianer Chen, Songjian Lu, Sing-Hoi Sze, and Fenghui Zhang.
\newblock Improved algorithms for path, matching, and packing problems.
\newblock In {\em Proceedings of the 17th Annual ACM-SIAM Symposium on Discrete
  Algorithms (SODA)}, pages 298--307. SIAM, 2007.

\bibitem{Chuzhoy16}
Julia Chuzhoy.
\newblock Improved bounds for the excluded grid theorem.
\newblock {\em CoRR}, abs/1602.02629, 2016.
\newblock URL: \url{http://arxiv.org/abs/1602.02629}.

\bibitem{CrowstonJMPRS13}
Robert Crowston, Mark Jones, Gabriele Muciaccia, Geevarghese Philip, Ashutosh
  Rai, and Saket Saurabh.
\newblock Polynomial kernels for lambda-extendible properties parameterized
  above the {Poljak-Turzik} bound.
\newblock In {\em {IARCS} Annual Conference on Foundations of Software
  Technology and Theoretical Computer Science ({FSTTCS})}, pages 43--54, 2013.
\newblock \href {http://dx.doi.org/10.4230/LIPIcs.FSTTCS.2013.43}
  {\path{doi:10.4230/LIPIcs.FSTTCS.2013.43}}.

\bibitem{cygan2015parameterized}
Marek Cygan, Fedor~V. Fomin, {\L}ukasz Kowalik, Daniel Lokshtanov, D{\'a}niel
  Marx, Marcin Pilipczuk, Micha{\l} Pilipczuk, and Saket Saurabh.
\newblock {\em Parameterized Algorithms}.
\newblock Springer, 2015.
\newblock \href {http://dx.doi.org/10.1007/978-3-319-21275-3}
  {\path{doi:10.1007/978-3-319-21275-3}}.

\bibitem{DemaineFHT05jacm}
Erik~D. Demaine, Fedor~V. Fomin, Mohammad~Taghi Hajiaghayi, and Dimitrios~M.
  Thilikos.
\newblock Subexponential parameterized algorithms on bounded-genus graphs and
  \emph{H}-minor-free graphs.
\newblock {\em Journal of the ACM}, 52(6):866--893, 2005.
\newblock \href {http://dx.doi.org/10.1145/1101821.1101823}
  {\path{doi:10.1145/1101821.1101823}}.

\bibitem{Diestel}
Reinhard Diestel.
\newblock {\em Graph theory}, volume 173 of {\em Graduate Texts in
  Mathematics}.
\newblock Springer-Verlag, Berlin, third edition, 2005.

\bibitem{FominK13}
Fedor~V. Fomin and Petteri Kaski.
\newblock Exact exponential algorithms.
\newblock {\em Communications of the ACM}, 56(3):80--88, 2013.
\newblock \href {http://dx.doi.org/10.1145/2428556.2428575}
  {\path{doi:10.1145/2428556.2428575}}.

\bibitem{FominLS14}
Fedor~V. Fomin, Daniel Lokshtanov, and Saket Saurabh.
\newblock Efficient computation of representative sets with applications in
  parameterized and exact algorithms.
\newblock In {\em Proceedings of the 24th Annual ACM-SIAM Symposium on Discrete
  Algorithms (SODA)}, pages 142--151, 2014.
\newblock \href {http://dx.doi.org/10.1137/1.9781611973402.10}
  {\path{doi:10.1137/1.9781611973402.10}}.

\bibitem{DBLP:journals/mst/GutinKLM11}
Gregory Gutin, Eun~Jung Kim, Michael Lampis, and Valia Mitsou.
\newblock Vertex cover problem parameterized above and below tight bounds.
\newblock {\em Theory of Computing Systems}, 48(2):402--410, 2011.
\newblock \href {http://dx.doi.org/10.1007/s00224-010-9262-y}
  {\path{doi:10.1007/s00224-010-9262-y}}.

\bibitem{GutinIMY12}
Gregory Gutin, Leo van Iersel, Matthias Mnich, and Anders Yeo.
\newblock Every ternary permutation constraint satisfaction problem
  parameterized above average has a kernel with a quadratic number of
  variables.
\newblock {\em Journal of Computer and System Sciences}, 78(1):151--163, 2012.
\newblock \href {http://dx.doi.org/10.1016/j.jcss.2011.01.004}
  {\path{doi:10.1016/j.jcss.2011.01.004}}.

\bibitem{HopcroftTarjan73}
John Hopcroft and Robert Tarjan.
\newblock Algorithm 447: Efficient algorithms for graph manipulation.
\newblock {\em Communications of the ACM}, 16(6):372--378, June 1973.
\newblock \href {http://dx.doi.org/10.1145/362248.362272}
  {\path{doi:10.1145/362248.362272}}.

\bibitem{HuffnerWZ08}
Falk H{\"{u}}ffner, Sebastian Wernicke, and Thomas Zichner.
\newblock Algorithm engineering for color-coding with applications to signaling
  pathway detection.
\newblock {\em Algorithmica}, 52(2):114--132, 2008.
\newblock \href {http://dx.doi.org/10.1007/s00453-007-9008-7}
  {\path{doi:10.1007/s00453-007-9008-7}}.

\bibitem{IP01}
Russell Impagliazzo and Ramamohan Paturi.
\newblock On the complexity of $k$-{SAT}.
\newblock {\em Journal of Computer and System Sciences}, 62(2):367--375, 2001.
\newblock \href {http://dx.doi.org/10.1006/jcss.2000.1727}
  {\path{doi:10.1006/jcss.2000.1727}}.

\bibitem{KawarabayashiK15}
Ken{-}ichi Kawarabayashi and Stephan Kreutzer.
\newblock The directed grid theorem.
\newblock In {\em Proceedings of the 47th Annual {ACM} Symposium on Theory of
  Computing (STOC)}, pages 655--664. {ACM}, 2015.
\newblock \href {http://dx.doi.org/10.1145/2746539.2746586}
  {\path{doi:10.1145/2746539.2746586}}.

\bibitem{KneisMRR06}
Joachim Kneis, Daniel M{\"{o}}lle, Stefan Richter, and Peter Rossmanith.
\newblock Divide-and-color.
\newblock In {\em Proceedings of the 32nd International Workshop on
  Graph-Theoretic Concepts in Computer Science ({WG})}, pages 58--67, 2006.
\newblock \href {http://dx.doi.org/10.1007/11917496_6}
  {\path{doi:10.1007/11917496_6}}.

\bibitem{Koutis08}
Ioannis Koutis.
\newblock Faster algebraic algorithms for path and packing problems.
\newblock In {\em Proceedings of the 35th International Colloquium on Automata,
  Languages and Programming (ICALP)}, volume 5125, pages 575--586. Springer,
  2008.
\newblock \href {http://dx.doi.org/10.1007/978-3-540-70575-8_47}
  {\path{doi:10.1007/978-3-540-70575-8_47}}.

\bibitem{KoutisW16}
Ioannis Koutis and Ryan Williams.
\newblock Algebraic fingerprints for faster algorithms.
\newblock {\em Communications of the ACM}, 59(1):98--105, 2016.
\newblock \href {http://dx.doi.org/10.1145/2742544}
  {\path{doi:10.1145/2742544}}.

\bibitem{LeafS15}
Alexander Leaf and Paul~D. Seymour.
\newblock Tree-width and planar minors.
\newblock {\em Journal of Combinatorial Theory, Series B}, 111:38--53, 2015.
\newblock \href {http://dx.doi.org/10.1016/j.jctb.2014.09.003}
  {\path{doi:10.1016/j.jctb.2014.09.003}}.

\bibitem{MahajanR99}
Meena Mahajan and Venkatesh Raman.
\newblock Parameterizing above guaranteed values: {MaxSat} and {MaxCut}.
\newblock {\em Journal of Algorithms}, 31(2):335--354, 1999.
\newblock \href {http://dx.doi.org/10.1006/jagm.1998.0996}
  {\path{doi:10.1006/jagm.1998.0996}}.

\bibitem{MahajanRS09}
Meena Mahajan, Venkatesh Raman, and Somnath Sikdar.
\newblock Parameterizing above or below guaranteed values.
\newblock {\em Journal of Computer and System Sciences}, 75(2):137--153, 2009.
\newblock \href {http://dx.doi.org/10.1016/j.jcss.2008.08.004}
  {\path{doi:10.1016/j.jcss.2008.08.004}}.

\bibitem{Monien85}
Burkhard Monien.
\newblock How to find long paths efficiently.
\newblock In {\em Analysis and design of algorithms for combinatorial problems
  ({U}dine, 1982)}, volume 109 of {\em North-Holland Math. Stud.}, pages
  239--254. North-Holland, Amsterdam, 1985.
\newblock \href {http://dx.doi.org/10.1016/S0304-0208(08)73110-4}
  {\path{doi:10.1016/S0304-0208(08)73110-4}}.

\bibitem{PapadimitriouY96}
Christos~H. Papadimitriou and Mihalis Yannakakis.
\newblock On limited nondeterminism and the complexity of the {V-C} dimension.
\newblock {\em Journal of Computer and System Sciences}, 53(2):161--170, 1996.
\newblock \href {http://dx.doi.org/10.1006/jcss.1996.0058}
  {\path{doi:10.1006/jcss.1996.0058}}.

\bibitem{RaymondT16}
Jean{-}Florent Raymond and Dimitrios~M. Thilikos.
\newblock Low polynomial exclusion of planar graph patterns.
\newblock {\em Journal of Graph Theory}, 84(1):26--44, 2017.
\newblock \href {http://dx.doi.org/10.1002/jgt.22009}
  {\path{doi:10.1002/jgt.22009}}.

\bibitem{RobertsonS-V}
Neil Robertson and Paul~D. Seymour.
\newblock Graph minors. {V}. {E}xcluding a planar graph.
\newblock {\em Journal of Combinatorial Theory, Series B}, 41(1):92--114, 1986.
\newblock \href {http://dx.doi.org/10.1016/0095-8956(86)90030-4}
  {\path{doi:10.1016/0095-8956(86)90030-4}}.

\bibitem{RobSeymT94}
Neil Robertson, Paul~D. Seymour, and Robin Thomas.
\newblock Quickly excluding a planar graph.
\newblock {\em Journal of Combinatorial Theory, Series B}, 62(2):323--348,
  1994.
\newblock \href {http://dx.doi.org/10.1006/jctb.1994.1073}
  {\path{doi:10.1006/jctb.1994.1073}}.

\bibitem{schrijver2002combinatorial}
Alexander Schrijver.
\newblock {\em Combinatorial Optimization: Polyhedra and Efficiency}.
\newblock Algorithms and Combinatorics. Springer Berlin Heidelberg, 2002.

\bibitem{Williams09}
Ryan Williams.
\newblock Finding paths of length $k$ in ${O}^*(2^k)$ time.
\newblock {\em Information Processing Letters}, 109(6):315--318, 2009.
\newblock \href {http://dx.doi.org/10.1016/j.ipl.2008.11.004}
  {\path{doi:10.1016/j.ipl.2008.11.004}}.

\bibitem{Zehavi14}
Meirav Zehavi.
\newblock Mixing color coding-related techniques.
\newblock In {\em Proceedings of the 23rd Annual European Symposium on
  Algorithms (ESA)}, volume 9294, pages 1037--1049. Springer, 2015.
\newblock \href {http://dx.doi.org/10.1007/978-3-662-48350-3_86}
  {\path{doi:10.1007/978-3-662-48350-3_86}}.

\end{thebibliography}

\appendix

\section{Unsatisfiability of the rerouting inequality systems}
\label{sec: LP unsatisfiable}

In this section, we verify manually that the three systems of linear inequalities established in the proof of Lemma~\ref{lem: detour-in-K4} are unsatisfiable, as claimed.
To this end, we interpret each system as a linear program and verify that its respective dual program is unbounded. From LP duality, it then follows that the primal program is infeasible.

\renewcommand{\vec}[1]{\mathbf{#1}}
\subsection{Preliminaries from LP theory}
Recall that a linear program $\mathcal{P}$ in standard form is specified by an objective vector $\vec c \in \Q^n$, a matrix $A \in \Q^{m \times n}$, and a bound vector $\vec b \in \Q^m$.
A vector $\vec x \in \Q^n$ is \emph{feasible} for $\mathcal{P}$ if $A \vec x
\leq \vec b$ holds coordinate-wise.  The program $\mathcal{P}$ then asks for a
feasible vector $\vec x \in \Q^n$ that maximizes the inner product $\vec c ^ T
\vec x$.
Such a vector does not necessarily exist since~$\mathcal{P}$ may fall into one
of the following degenerate cases:
\begin{itemize}
\item There may be no feasible vector for $\mathcal{P}$ at all. In this case, we say that $\mathcal{P}$ is \emph{infeasible}.
\item The function $\vec c^T \vec x$ may attain arbitrarily large values for feasible $\vec x$. In this case, we say that $\mathcal{P}$ is \emph{unbounded}.
\end{itemize}
These two degenerate cases are intimately linked by the theory of LP duality:
The dual of~$\mathcal{P}$, denoted by $\mathcal{D}(\mathcal{P})$, asks for a vector $\vec y \in \Q^m $ that minimizes $\vec b ^ T \vec y$ subject to $A^T \vec y = \vec c$ and $\vec y \geq 0$. It is known that, if $\mathcal{D}(\mathcal{P})$ is unbounded, then $\mathcal{P}$ is infeasible \cite{schrijver2002combinatorial}.

\subsection{Proving unboundedness of the duals} Given a system of linear inequalities $A \vec
x\leq \vec b$, we define a linear program $\mathcal{P}$ by endowing the system
with the objective vector $\vec c = \vec 0$. We show that $\mathcal{P}$ is
infeasible by proving $\mathcal{D}(\mathcal{P})$ to be unbounded. To do so, it
suffices to exhibit a vector $\vec{y^*}$ with $\vec{y^*} \geq \vec 0$ and
$A^T \vec{y^*} = \vec 0$ and $\vec b ^ T \vec{y^*} < 0$. The multiples $\alpha
\cdot \vec{y^*}$ with $\alpha > 0$ are then feasible as well, and they attain
arbitrarily small objective values; thus, $\vec{y^*}$ is a witness of the fact
that $\mathcal{D}(\mathcal{P})$ is unbounded.

To find such vectors $\vec{y^*}$, we take a closer look at the systems of
inequalities that appear in the proof of Lemma~\ref{lem: detour-in-K4}.
We transform these inequalities into normal form, and observe that $\vec b$ has
non-zero entries only at rows generated by the inequalities (\ref{eq: no long
  path}) that act as upper bounds on the path lengths.  Furthermore, these entries are all equal to~$-1$.
To prove infeasibility of the primal program, it thus suffices to find a
vector $\vec{y^*}\ge 0$ that assigns a non-zero value to at least one variable corresponding
to a primal constraint from~(\ref{eq: no long path}).

For each of the three cases in Lemma~\ref{lem: detour-in-K4}, we exhibit such
vectors $\vec{y^*}$.
To improve legibility, we display these vectors as linear combinations of the
involved rows and omit rows whose corresponding coefficient is zero.
Furthermore, we abbreviate expressions like $\ell(b_1 u)$ to $\ell_{1,u}$.
Note that each of the three listed linear combinations indeed
\begin{itemize}
\item involves only inequalities from the respective case,
\item is a feasible solution to the dual because it evaluates to the zero vector
  and involves only non-negative coefficients, and
\item proves the unboundedness of the dual because it assigns a positive
  coefficient to some row generated by the set of inequalities (\ref{eq: no long
    path}) in the respective case.
\end{itemize}

We found these solutions
in a bleary-eyed state using the computer algebra system MATLAB, but this is
irrelevant, as their correctness can be verified immediately by hand.

\paragraph*{Case \ref{uv a}:}
\[
\begin{array}{rrrrrrrrrrrrr}
  \vec{y^*}\cdot A= & 2 & \cdot \quad ( & \textcolor{lightgray}{0} & k &
  \textcolor{lightgray}{0}  & \textcolor{lightgray}{0} &
  \textcolor{lightgray}{0} & \textcolor{lightgray}{0} & \textcolor{lightgray}{0}
  & \textcolor{lightgray}{0} & - \ell_{3,4} & \enspace )\\
  & +1 & \cdot \quad ( & - d & - k & \ell_{1,u}  & \ell_{2,v} &
  \textcolor{lightgray}{0} & \ell_{1,4} & \ell_{2,3} & \textcolor{lightgray}{0}
  & \ell_{3,4} & \enspace )\\
  & +1 & \cdot \quad ( & - d & - k & \ell_{1,u}  & \ell_{2,v} & \ell_{1,3} &
  \textcolor{lightgray}{0} & \textcolor{lightgray}{0} & \ell_{2,4} & \ell_{3,4}
  & \enspace )\\
  & +1 & \cdot \quad ( & d & \textcolor{lightgray}{0} & - \ell_{1,u}  & -
  \ell_{2,v} & - \ell_{1,3} & \textcolor{lightgray}{0} & - \ell_{2,3} &
  \textcolor{lightgray}{0} & \textcolor{lightgray}{0} & \enspace )\\
  & +1 & \cdot \quad ( & d & \textcolor{lightgray}{0} & - \ell_{1,u}  & -
  \ell_{2,v} & \textcolor{lightgray}{0} & - \ell_{1,4} &
  \textcolor{lightgray}{0} & - \ell_{2,4} & \textcolor{lightgray}{0} & \enspace
  )
\end{array}
\]
For clarity, we remark that $\vec{y^*}$ has five non-zero entries: one is~$2$
and four are~$1$.
The dual variable set to~$2$ in the first line corresponds to
the equation $\ell(b_3b_4)\ge k$, which is a constraint from~\eqref{eq:
  tetrasubdiv}.
The second and third line correspond to the
constraints of type \eqref{eq: no long path} for the paths~$ub_1b_4b_3b_2v$ and~$ub_1b_3b_4b_2v$.
The fourth and fifth line correspond to the constraints of type \eqref{eq: shortest path}
for the paths~$ub_1b_3b_2v$ and~$ub_1b_4b_2v$.

Finally, note that $\vec{y^*}\cdot A=0$ and $\vec b^T\vec{y^*} = -2$ hold.
The latter follows since the entries of~$\vec b$ are equal to~$-1$ at
inequalities~\eqref{eq: no long path}, and zero otherwise.

\paragraph*{Case \ref{uv b}:}
\[
\begin{array}{rrrrrrrrrrrrrrr}
  &\vec{y^*}\cdot A = \\
  & 2 &\cdot \quad ( & \textcolor{lightgray}{0} & k & \textcolor{lightgray}{0} &
  \textcolor{lightgray}{0} & \textcolor{lightgray}{0} & \textcolor{lightgray}{0}
  & \textcolor{lightgray}{0} & - \ell_{2,3} & \textcolor{lightgray}{0} &
  \textcolor{lightgray}{0} & \enspace )\\
  & +1 & \cdot \quad ( & - d & - k & \ell_{1,u} & \textcolor{lightgray}{0} &
  \textcolor{lightgray}{0} & \ell_{3,v} & \ell_{1,4} & \ell_{2,3} & \ell_{2,4} &
  \textcolor{lightgray}{0} & \enspace )\\
  & +1 & \cdot \quad ( & - d & - k & \textcolor{lightgray}{0} & \ell_{2,u} &
  \ell_{1,v} & \textcolor{lightgray}{0} & \ell_{1,4} & \ell_{2,3} &
  \textcolor{lightgray}{0} & \ell_{3,4} & \enspace )\\
  & +1 & \cdot \quad ( & d & \textcolor{lightgray}{0} & \textcolor{lightgray}{0}
  & - \ell_{2,u} & - \ell_{1,v} & \textcolor{lightgray}{0} & - \ell_{1,4} &
  \textcolor{lightgray}{0} & - \ell_{2,4} & \textcolor{lightgray}{0} & \enspace
  )\\
  & +1 & \cdot \quad ( & d & \textcolor{lightgray}{0} & - \ell_{1,u} &
  \textcolor{lightgray}{0} & \textcolor{lightgray}{0} & - \ell_{3,v} & -
  \ell_{1,4} & \textcolor{lightgray}{0} & \textcolor{lightgray}{0} & -
  \ell_{3,4} & \enspace )
\end{array}
\]

\paragraph*{Case \ref{uv c}:}
\[
\begin{array}{rrrrrrrrrrrrr}
  \vec{y^*}\cdot A = & 2 & \cdot \quad ( & \textcolor{lightgray}{0} & k &
  \textcolor{lightgray}{0} & \textcolor{lightgray}{0} & \textcolor{lightgray}{0}
  & \textcolor{lightgray}{0} & - \ell_{1,3} & \textcolor{lightgray}{0} &
  \textcolor{lightgray}{0}  & \enspace )\\
  & +1 & \cdot \quad ( & - d & - k & \textcolor{lightgray}{0} & \ell_{2,u} &
  \textcolor{lightgray}{0} & \ell_{4,v} & \ell_{1,3} & \ell_{1,4} & \ell_{2,3}
  & \enspace )\\
  & +1 & \cdot \quad ( & - d & - k & \ell_{1,u} & \textcolor{lightgray}{0} &
  \ell_{3,v} & \textcolor{lightgray}{0} & \ell_{1,3} & \textcolor{lightgray}{0}
  & \textcolor{lightgray}{0}  & \enspace )\\
  & +1 & \cdot \quad ( & d & \textcolor{lightgray}{0} & - \ell_{1,u} &
  \textcolor{lightgray}{0} & \textcolor{lightgray}{0} & - \ell_{4,v} &
  \textcolor{lightgray}{0} & - \ell_{1,4} & \textcolor{lightgray}{0}  & \enspace
  )\\
  & +1 & \cdot \quad ( & d & \textcolor{lightgray}{0} & \textcolor{lightgray}{0}
  & - \ell_{2,u} & - \ell_{3,v} & \textcolor{lightgray}{0} &
  \textcolor{lightgray}{0} & \textcolor{lightgray}{0} & - \ell_{2,3}  & \enspace
  )
\end{array}
\]

\end{document}